\documentclass{article}
\usepackage{a4wide}
\usepackage{euler}
\usepackage{etex}
\usepackage{enumerate}

\usepackage{amsmath,amssymb}
\usepackage{amsthm}
\usepackage{setspace}
\usepackage[all]{xy}
\usepackage{multirow}

\usepackage{todonotes}
\newcommand{\todoi}[1]{{}}
\newcommand{\todox}[1]{{}}

\usepackage{soul}
\usepackage{lscape}
\newtheorem{theorem}{Theorem} 
\renewenvironment{proof}{{\noindent\underline{Proof}.}}{ {\ \hfill $\Box$}}

\theoremstyle{definition}
\newtheorem{problem}{Problem} 
\newtheorem{example}{Example} 
\newtheorem{observation*}{Observation} 
\newtheorem{definition}{Definition} 
\theoremstyle{plain}

\newtheorem{assumption}{Assumption} 
\newtheorem{property}{Property} 
\newtheorem{conjecture}{Conjecture}

\usepackage{xfrac}
\usepackage{tikz}
\usepackage{arydshln}

\newcommand{\Mcomma}{\text{,}}
\newcommand{\Mdot}{\text{.}}

\newcommand{\close}{\ \hfill $\Box$}
\newcommand{\closex}{$\:\Box$}

\newcommand{\addr}[1]{{\footnotesize 
  $\tt [$\url{#1}$\tt ]$}}

\newcommand{\fnsz}{\footnotesize}

\newcommand{\ene}{\mathbb{N}}

\newcommand{\defined}{\stackrel{\text{\tiny\rm def}}{=}}
\newcommand{\assign}{\leftarrow}

\newcommand{\pair}[1]{\langle #1 \rangle}

\newcommand{\cod}[1]{\text{\rm cod}(#1)}
\newcommand{\pc}[1]{\text{\rm cplt}(#1)}

\newcommand{\CFG}{\text{\rm CFG}}

\newcommand{\fxg}{{\sf g}}

\newcommand{\fxh}{{\sf h}}
\newcommand{\codg}{{\rm cod({\sf g})}}

\newcommand{\fINP}{{\rm INP}-$f$}
\newcommand{\fOUT}{{\rm OUT}-$f$}
\newcommand{\NEITHER}{{\rm NEITHER}}

\newcommand{\tfxg}{$\sf g$}
\newcommand{\tfxgi}[1]{${\sf g}_{#1}$}
\newcommand{\tfxh}{$\sf h$}

\newcommand{\lamb}[2]{\lambda #1\,[#2]}

\newcommand{\spac}{\hspace*{5mm}}
\newcommand{\spacc}{\hspace*{16mm}}

\newcommand{\inst}{\spac{\rm Instance:}\ }
\newcommand{\quest}{\spac{\rm Question:}\ }
\newcommand{\class}{\spac{\rm Class:}\ }

\newcommand{\lra}{\Leftrightarrow}
\newcommand{\Rarr}{\Rightarrow}
\newcommand{\halts}{\!\!\downarrow}

\newcommand{\eqnref}[1]{(\ref{#1})}

\newcommand{\und}[1]{\underline{#1}}
\newcommand{\ov}[1]{{\overline{#1}}}

\newcommand{\Smn}{{${\rm S}_{\rm mn}$}}

\newcommand{\sigZero}{\Sigma_0}
\newcommand{\piZero}{\Pi_0}
\newcommand{\deltaZero}{\Delta_0}
\newcommand{\deltaOne}{\Delta_1}
\newcommand{\deltaTwo}{\Delta_2}
\newcommand{\deltaThree}{\Delta_3}

\newcommand{\deltaN}{\Delta_n}
\newcommand{\sigOne}{\Sigma_1}
\newcommand{\sigTwo}{\Sigma_2}
\newcommand{\sigThree}{\Sigma_3}

\newcommand{\sigN}{\Sigma_n}
\newcommand{\sigM}{\Sigma_m}
\newcommand{\sigNO}{\Sigma_{n-1}}
\newcommand{\piOne}{\Pi_1}
\newcommand{\piTwo}{\Pi_2}

\newcommand{\piN}{\Pi_n}
\newcommand{\piM}{\Pi_m}
\newcommand{\piNO}{\Pi_{n-1}}

\newcommand{\deltaM}{\Delta_m}

\newcommand{\prinj}{\text{\normalfont INJECTIVE}^{\text{PR}}}
\newcommand{\prinjN}{$\text{\normalfont INJECTIVE}^{\text{PR}}$}
\newcommand{\Gprinj}{\text{\normalfont INJECTIVE}}
\newcommand{\GprinjN}{$\text{\normalfont INJECTIVE}$}

\newcommand{\prsurj}{\text{\normalfont ONTO}^{\text{PR}}}
\newcommand{\prsurjN}{$\text{\normalfont ONTO}^{\text{PR}}$}
\newcommand{\Gprsurj}{\text{\normalfont ONTO}}
\newcommand{\GprsurjN}{$\text{\normalfont ONTO}$}

\newcommand{\prbij}{\text{\normalfont BIJECTIVE}^{\text{PR}}}
\newcommand{\prbijN}{$\text{\normalfont BIJECTIVE}^{\text{PR}}$}
\newcommand{\Gprbij}{\text{\normalfont BIJECTIVE}}
\newcommand{\GprbijN}{$\text{\normalfont BIJECTIVE}$}

\newcommand{\ang}[1]{{\langle #1 \rangle}}
\newcommand{\imply}{\Rightarrow}
\newcommand{\vph}{\varphi}
\newcommand{\stat}{{\sf St}}

\newcommand{\PRversion}[1]{{#1}^{\rm PR}}

\newcommand{\abs}[1]{{\lvert#1\rvert}}

\newcommand{\reduces}{{\,\leq_m\,}}

\newcommand{\HP}{\text{\rm HP}}
\newcommand{\SHP}{\text{\rm SHP}}
\newcommand{\total}{{\rm TOTAL}}

\newcommand{\totalN}{\text{\rm TOTAL}}
\newcommand{\fdom}{\text{\normalfont FINITE-DOMAIN}}
\newcommand{\fdomN}{$\text{\normalfont FINITE-DOMAIN}$}

\newcommand{\cfdomN}{$\text{\normalfont COFINITE-DOMAIN}$}

\newcommand{\recur}{\text{\normalfont RECURSIVE}}
\newcommand{\recurN}{$\text{\normalfont RECURSIVE}$}
\newcommand{\recurPR}{\text{\normalfont RECURSIVE}^{\text{\rm PR}}}
\newcommand{\recurPRN}{$\text{\normalfont RECURSIVE}^{\text{\rm PR}}$}

\newcommand{\PRpow}[1]{{#1}^{\text{\rm PR}}}

\newcommand{\GozN}{$\text{\normalfont HAS-ZEROS}$}
\newcommand{\oz}{\text{\normalfont HAS-ZEROS}^{\text{\rm PR}}}
\newcommand{\ozN}{$\text{\normalfont HAS-ZEROS}^{\text{\rm PR}}$}

\newcommand{\eozN}{$\text{\normalfont EXACTLY-ONE-ZERO}^{\text{\rm PR}}$}
\newcommand{\Geoz}{\text{\normalfont EXACTLY-ONE-ZERO}}
\newcommand{\GeozN}{$\text{\normalfont EXACTLY-ONE-ZERO}$}

\newcommand{\ozFG}{\text{\rm HAS-ZEROS-{$f$}$\cdot {\sf g}^{\text{\rm PR}}$}}

\newcommand{\ozHF}{\text{\rm HAS-ZEROS-${\sf h}\cdot${$f^{\text{\rm PR}}$}}}

\newcommand{\ozHFG}{\text{\rm HAS-ZEROS-${\sf h}\cdot${$f$}$\cdot {\sf g}^{
                    \text{\rm PR}}$}}

\newcommand{\GnoZerosN}{$\text{\normalfont NO-ZEROS}$}

\newcommand{\gtkzN}{$\text{\normalfont AT-LEAST-$k$-ZEROS}^{\text{\rm PR}}$}
\newcommand{\Ggtkz}{\text{\normalfont AT-LEAST-$k$-ZEROS}}
\newcommand{\GgtkzN}{$\text{\normalfont AT-LEAST-$k$-ZEROS}$}

\newcommand{\ekzN}{$\text{\normalfont EXACTLY-$k$-ZEROS}^{\text{\rm PR}}$}

\newcommand{\GekzN}{$\text{\normalfont EXACTLY-$k$-ZEROS}$}

\newcommand{\az}{\text{\normalfont ZERO-FUNCTION}^{\text{\rm PR}}}
\newcommand{\azN}{$\text{\normalfont ZERO-FUNCTION}^{\text{\rm PR}}$}

\newcommand{\GazN}{$\text{\normalfont ZERO-FUNCTION}$}

\newcommand{\eqNXT}{\text{\normalfont EQUAL-NEXT}^{\text{\rm PR}}}
\newcommand{\eqNXTN}{$\text{\normalfont EQUAL-NEXT}^{\text{\rm PR}}$}

\newcommand{\infzN}{$\text{\normalfont $\infty$-ZEROS}^{\text{\rm PR}}$}

\newcommand{\GinfzN}{$\text{\normalfont $\infty$-ZEROS}$}

\newcommand{\aazN}{$\text{\normalfont ALMOST-ALL-ZEROS}^{\text{\rm PR}}$}

\newcommand{\equiN}{$\text{\normalfont EQUIVALENCE}^{\text{\rm PR}}$}

\newcommand{\GequiN}{$\text{\normalfont EQUIVALENCE}$}

\newcommand{\kCodN}{$\text{$\lvert$\normalfont CODOMAIN$\rvert\!=\!k$}^{
       \text{\rm PR}}$}

\newcommand{\kCodON}{$\text{$\lvert$\normalfont CODOMAIN$\rvert\!=\!1$}^{
       \text{\rm PR}}$}

\newcommand{\GkCodN}{$\text{$\lvert$\normalfont CODOMAIN$\rvert\!=\!k$}$}

\newcommand{\GkCodON}{$\text{$\lvert$\normalfont CODOMAIN$\rvert\!=\!1$}$}
\newcommand{\GkCodTN}{$\text{$\lvert$\normalfont CODOMAIN$\rvert\!=\!2$}$}

\newcommand{\GkCodL}{$\text{$\lvert$CODOMAIN$\rvert$}$}

\newcommand{\fCodN}{$\text{\normalfont FINITE-CODOMAIN}^{\text{\rm PR}}$}

\newcommand{\infCodN}{$\text{\normalfont INFINITE-CODOMAIN}^{\text{\rm PR}}$}

\newcommand{\zeroMoreN}{$\text{\normalfont HAS-ZEROS-AND-NONZEROS}^{
            \text{\rm PR}}$}

\newcommand{\eOneN}{$\text{\rm EQUAL-AT-ONE-POINT}^{\text{\rm PR}}$}

\newcommand{\FFN}{$\text{{$ff$}}^{\text{\rm PR}}$}

\newcommand{\FFtZN}{$\text{$ff${\sc Z}{\normalfont 2}}^{\text{\rm PR}}$}

\newcommand{\FFnN}{$\text{$(f^{(n)})$}^{\text{\rm PR}}$}

\newcommand{\FFRN}[1]{$(\text{${f}^{(#1)}$})^{\text{\rm PR}}$}

\newcommand*{\dminus}{%
\mathrel{\vcenter{\offinterlineskip
\hbox{$\hspace*{1.15mm}\cdot$}\vskip-.8ex\hbox{$\hspace*{0.24mm}-$}}}}

\definecolor{dgreen}{rgb}{0.0,0.6,0.0}
\definecolor{dbrown}{rgb}{0.3,0.1,0.0}
\definecolor{dred}{rgb}{0.6,0.2,0.0}
\definecolor{ddred}{rgb}{0.6,0.4,0.2}
\definecolor{dgreen}{rgb}{0.1,0.6,0.2}
\newcommand{\eiroi}[1]{${\rm #1}^{\text{\underline{\rm o}}}$}
\newcommand{\IPR}{{\rm I}_{\rm PR}}
\newcommand{\vertSymb}[2]{%
  \mathrel{\raisebox{#2}{#1}}
}

\usepackage[hang,flushmargin]{footmisc}

\begin{document}

\begin{center}
{\Large \bf Primitive recursive functions versus\\
        partial recursive functions:\\
   \vspace*{2.5mm}
        comparing the degree of undecidability}

\bigskip
{\large Armando B. Matos\footnote{Address:
    Rua da Venezuela 146, \eiroi{1} Dto, 4150-743, Porto, Portugal.\\
    Affiliation: LIACC, Artificial Intelligence and Computer Science Laboratory,
    Universidade do Porto,
    Rua do Campo Alegre, 4169-007, Porto, Portugal.\\
    email: {\tt armandobcm@yahoo.com}\\
    LIACC homepage: {\tt www.liacc.up.pt}}}

\bigskip

{\large July 2006}

\end{center}

\vfill

\begin{abstract}
  Consider a decision problem whose instance is a function.
  Its degree of undecidability, measured by 
  the corresponding class of the arithmetic (or Kleene-Mostowski) hierarchy
  hierarchy, may depend on whether the instance is a partial recursive
  or a primitive recursive function. 
  A similar situation happens for results
  like Rice Theorem (which is false for primitive recursive functions).
  Classical Recursion Theory deals mainly with the properties
  of partial recursive functions.

  We study several natural decision problems related 
  to {\em primitive recursive functions} and characterise their degree of 
  undecidability.
  As an example, we show that, for primitive recursive functions,
  the injectivity problem is~$\Pi^0_1$-complete while the surjectivity 
  problem is $\Pi^0_2$-complete.

  We compare the degree of undecidability (measured by the level in the
  arithmetic hierarchy) of several primitive recursive decision problems 
  with the corresponding problems of classical Recursion Theory.
  For instance, the problem ``does the codomain of a function
  have exactly one element?'' is $\piOne$-complete
  for primitive recursive functions and belongs
  to the class $\deltaTwo\setminus(\sigOne\cup\piOne)$
  for partial recursive functions.

  An important decision problem, ``does a
  given primitive recursive function have at least one zero?'' 
  is studied in detail; the input and output restrictions that are
  necessary and sufficient for the decidability this problem~-- its 
  ``frontiers of decidability''~-- are established.
  We also study a more general
  situation in which a primitive recursive function (the instance of the
  problem) is a part of an arbitrary ``acyclic primitive
  recursive function graph''. This setting may be
  useful to evaluate the relevance of a given primitive recursive
  function as a part of a larger primitive recursive structure.
\end{abstract}


\vfill

\newpage
\ 

\vfill

\tableofcontents

\vfill

\newpage

\section{Introduction}
The primitive recursive (PR) functions form a large and important
enumerable subclass of the recursive (total computable)
functions.  Its vastness is obvious from the fact that every total
recursive function whose time complexity is bounded by a primitive
recursive function is itself primitive recursive, while its importance
is well expressed by the Kleene Normal Form Theorem (see
\cite{kleene52,BBJ,odi}) and by the fact that every recursively
enumerable set can be enumerated by a primitive recursive
function\footnote{In the sense that, if~$L$ is recursively enumerable,
  there is a primitive recursive function~$f$ such that $L\cup
  \{0\}=\{f(i):i\in\ene\}$.}.

The meta-mathematical and algebraic properties of PR functions have
been widely studied, see for instance
\cite{kleene52,Robinson00,JRob2,ISzal1,RRob2}.  Special sub-classes (such
as unary PR functions, see for instance \cite{Robinson68,Sever1}) and
hierarchies (see for instance
\cite{grz,axt59,MeyerRitchie67a,moll,BN,GoeNeh2}
and~\cite[Ch.~VI]{yasuhara}) of PR functions were also analysed.

Efficiency lower bounds of primitive recursive
algorithms have also been studied,
see~\cite{colson,CF98,CLV,david01,moll,moschovakis,fredholm,dries03,
  valarcher,matosPRF}.

In this work we concentrate on decidability questions related to PR
functions.

Many decision problems related to PR functions are
undecidable. For instance, the question ``given a PR
function~$f(x)$, is there some~$x$ such
that~$f(x)=0$?'' is undecidable~(\cite{BBJ}). In this paper we show
that many other questions related to PR functions are undecidable, and
classify their degree of undecidability according to the arithmetic
hierarchy.  Let us mention two examples of the problems studied in
this work: ``does a given PR function have exactly one zero?''
(class $\deltaTwo\setminus (\sigOne\cup\piOne)$), and
``is a given PR function surjective (onto)?'' ($\piTwo$-complete).

The undecidability of some of these problems has undesirable
consequences. For instance, when studying reversible computations, it
would be useful to have a method for enumerating the bijective PR
functions. However, although the set of
PR functions is recursively enumerable, the set of bijective PR functions is not:
it belongs the $\piTwo$-complete class.

We also study classes of PR problems in which a PR function~$f$ is a
part of a larger PR system~$S$. We obtain a simple decidability
condition for the case in which 
\begin{equation}
  \label{fghFORM}
  S(f,\ov{x})=\fxh(f(\fxg(\ov{x})))\Mcomma
\end{equation}
where~$\fxg$ and~$\fxh$ are fixed (total) recursive functions 
($\fxg$ may have multiple
outputs), but not for the case illustrated in Figure~\ref{inside} 
(page~\pageref{inside}).
It is interesting to compare~\eqnref{fghFORM} with the result of applying
the Normal Form Theorem to this more general case  (see page~\pageref{nft})
\begin{equation}
  \label{acyclicFORM}
  S(f,\ov{x})=\fxh(\ov{x},f(\fxg(\ov{x})))\Mdot
\end{equation}

\subsection*{Organisation and contents of the paper.} 
The main concepts and results needed for the rest of the paper are
presented in Section~\ref{prel} (page~\pageref{prel}).  
The rest of the paper can be divided in two parts.
\begin{figure}[t]
  \begin{center}
  \begin{tikzpicture}[xshift=2.5cm, yshift=2.0cm, scale=1.8, auto,
    main node/.style={circle,fill=blue!20,draw}]
    \draw [red, very thick, fill=yellow!05] (0.5,1) 
           rectangle (3.5,3);
    \draw [->,thick]     (0.0,2.10) -- (0.5,2.10);
    \draw [thick,dotted] (0.0,2.00) -- (0.5,2.00);
    \draw [->,thick]     (0.0,1.90) -- (0.5,1.90);
    \node at (-0.2,2)     {$\ov{x}$};
    \draw [->,thick]     (3.5,2) -- (4.0,2);
    \node at (4.2,2)     {$y$};
    \node at (3.2,1.25)   {$\LARGE S$};
    \draw [thin, fill=blue!20] (2,2) circle (0.2);
    \node at (2,2) {$f$};
    \draw [->,thick]     (1.7,2.10) -- (1.83,2.10);
    \draw [thick,dotted] (1.2,2.00) -- (1.80,2.00);
    \draw [->,thick]     (1.7,1.90) -- (1.83,1.90);
    \draw [thick]        (1.7,2.10) -- (1.40,2.5);
    \draw [thick]        (1.7,1.90) -- (1.40,1.5);
    \draw [thick] (2.2,2) -- (2.3,2);
    \draw [thick] (2.3,2) circle (0.01);
    \draw [->,thick] (2.3,2) -- (2.9,2);
    \draw [->,thick] (2.3,2) -- (2.6,2.5);
    \draw [->,thick] (2.3,2) -- (2.6,1.5);
    \draw [->,thick,dotted] (1.40,2.50) -- (2.10,2.50);
    \draw [thick,dotted]    (2.10,2.50) -- (2.30,2.50);
    \draw [->,thick,dotted] (1.40,1.50) -- (2.10,1.50);
    \draw [thick,dotted]    (2.10,1.50) -- (2.30,1.50);
  \end{tikzpicture}
  \end{center} 
  \caption{A PR function~$f$ may be considered an important
    part of a complex structure~$S$ if the following problem is
    undecidable: \und{instance:}~$f$, \und{\smash{question:}} 
    ``$\exists\ov{x}:S(f,\ov{x})=0$?''. 
    We will show that an arbitrary acyclic
    system~$S$ containing exactly one occurrence of~$f$ can always be 
    reduced to the normal form illustrated in Figure~\ref{nf}
    (page~\pageref{nf}).}
\label{inside}
\end{figure}
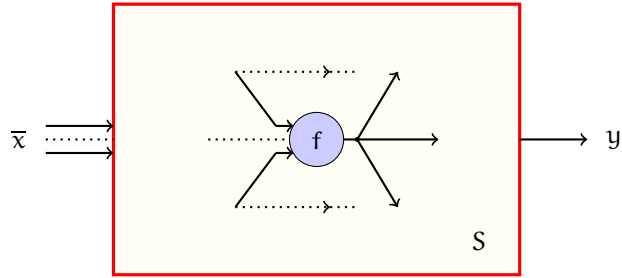

\subsubsection*{Part I, specific problems}
In Section~\ref{problems} (page~\pageref{problems}) several natural decision
problems associated with PR functions are studied and compared
(in terms of the arithmetic hierarchy level) with the corresponding partial 
recursive problems.
The list of problem includes: 
the existence of zeros, the equivalence of functions, ``domain
problems'', and problems related with the injectivity, surjectivity,
and bijectivity of PR functions. For each of these problems the degree
of undecidability is established in terms of their exact location in
the in the arithmetic hierarchy; see Figures~\ref{summary}
(page~\pageref{summary}) and~\ref{classes} (page~\pageref{classes}).
The problems we study are either recursive,
complete\footnote{Many natural decision problems seem to be complete in~$\sigN$ 
or in~$\piN$, see for instance the examples discussed in ``Specific index sets'',
Section~X.9 of~\cite{odiTwo} and also the problems studied in~\cite{simonsen}.}
in~$\sigN$ (for some $n\geq 1$), complete in~$\piN$ (for some $n\geq 1$)
or belong to the class $\deltaTwo\setminus(\sigOne\cup\piOne)$.


\subsubsection*{Part II, classes of problems}
\begin{enumerate}
\item [--] Section~\ref{composition} (page~\pageref{composition}):
  study of the ``frontiers of decidability'' of the problem
  ``does a given PR function has at least one zero?''. More precisely,
  we find the input and output restrictions that are necessary and
  sufficient for decidability; these restrictions are obtained by pre-
  or post-applying fixed PR functions to the PR function which is the
  instance of the problem.  
\item [--] Section~\ref{geni} (page~\pageref{geni}):
  two more general settings are studied:
  \begin{enumerate}
  \item [--] Section~\ref{plug} (page~\pageref{plug}):
    the PR function~$f$, instance of a decision
    problem, is located somewhere inside a larger acyclic PR
    system~$S$ (also denoted by $S(f,\ov{x})$), see
    Figure~\ref{inside} (page~\pageref{inside}).  In algebraic terms
    this means that~$S$ is obtained by {\em composing} fixed PR
    functions (or, more generally, fixed recursive functions) 
    with the instance~$f$. An example of such a system~$S$
    can be seen in see Figure~\ref{ag} (page~\pageref{ag}).
    The class of problems, parameterised by~$S$ is:
    given a PR function~$f$, is there some~$\ov{x}$ such that 
    $S(f,\ov{x})=0$?
    When this problem is undecidable the PR function~$f$ can be considered 
    an ``important part'' of~$S$.

  \item[--]  Section~\ref{itself} (page~\pageref{itself}):
    The problem of composing the PR function~$f$ (the instance of the
    problem) with itself is briefly studied.
  \end{enumerate}
\item [--] In Section~\ref{PR-ParRec} (page~\pageref{PR-ParRec}) we
  compare the degree of undecidability of the partial recursive and primitive
  recursive versions of the same decision problem,
  see Figure~\ref{four}, page~\pageref{four}.
  For $n\geq 2$ assign the level~$n-\sfrac{1}{2}$ to the class
  $\deltaN\setminus(\sigNO\cup\piNO)$.
  We conjecture that
  there exist problems whose difference of levels in the arithmetic hierarchy 
  is~0, 1, $\sfrac{3}{2}$, 2, $\sfrac{5}{2}$, 3, $\sfrac{7}{2}$\ldots\ 
  (See Figure~\ref{discreps}, page~\pageref{discreps}). 
  In the general case it may be difficult or impossible to determine that 
  difference of levels because the PR level is influenced by the 
  particularities of the PF functions.
\end{enumerate}

\todoi{Removed the ``motivation'' related to a possible PR Rice Theorem.}

\section{Preliminaries}
\label{prel}
The notation used in this paper is straightforward.  ``Iff'' denotes
``if and only if'' and~``$\defined$'' denotes ``is defined as''.  The
logical negation is denoted by~``$\neg$''. The complement of the
set~$A$ is denoted by~$\ov{A}$, while the difference between the
sets~$A$ and~$B$ is denoted by~$A\setminus B$.  
A total function $f:A\to B$ is {\em injective} (or {\em one-to-one})
if~$x\neq y$ implies~$f(x)\neq
f(y)$; it is {\em surjective} (or {\em onto}) if for every~$y\in B$
there is at least one~$x\in A$ such that~$f(x)=y$; and it is {\em
  bijective} if it is both injective and surjective.
The set $\{x:f(x)=y\}$ is denoted by $f^{-1}(y)$.
For~$n\geq 0$, the sequence of arguments
``$x_1,\,x_2,\,\ldots,\,x_n$'' is denoted by~$\ov{x}$. 
Using a standard primitive recursive bijection $\ang{x_1,x_2}$ 
of~$\ene^2$ in~$\ene$ we can represent a sequence~$\ov{x}$ of integers
$x_1,\,x_2,\,\ldots,\,x_n$ by a single integer
$x=\ang{x_1,\,x_2,\,\ldots,\,x_n}
\defined \ang{x_1,\ang{x_2,\ang{\ldots\ang{x_{n-1},x_n}\ldots}}}$;
the sequence~$\ov{x}$ will often be denoted by~$x$.

\begin{definition}
  \label{multif}
  Let $\ov{x}=\langle x_1,\ldots,x_n\rangle$ and 
  $\ov{y}=\langle y_1,\ldots,y_m\rangle$; 
  by $\ov{y}=f(\ov{x})$ we mean that the tuple~$\ov{y}$
  is a function of the tuple~$\ov{x}$.
  We say that~$f$ is a {\em multiple function} or ``multifunction''.
\close
\end{definition}
The \Smn\ Theorem~\cite{kleene52,davis,Rogers} will often be used without
mention.

\subsection{Languages, reductions and the arithmetic hierarchy}
\label{langs}

\bigskip

\todoi{Time-limited Turing machines.}
\vspace*{-3mm}
We say that the functions~$f$ and~$g$ are {\em equal} and write $f=g$
if for every integer~$x$ either~$f(x)$ and~$g(x)$ are both undefined,
or they are both defined and have the same value.

\todox{Relate with Definition \ref{interval}.}
The terms ``recursive function'', ``total computable function'',
and ``total function'' will be used interchangeably.
PR denotes ``primitive recursive'', ParRec denotes ``partial recursive''. 
The ParRec function corresponding to the index~$e$ is~$\vph_e$.
Its domain is denoted by~$W_e$. ``Recursively enumerable'' is abbreviated 
by r.e.

\subsubsection*{Injectivity, surjectivity, bijectivity} 
\label{ParRecInj}
PR functions are total (computable) so that the definition at the beginning
of Section~\ref{prel} apply. We use the following definition for ParRec
functions.

\begin{definition}
\label{Ginjsurbij}
Let~$\vph_e$ be a partial recursive function. \\
\spac $\vph_e$ is injective if its restriction to~$W_e$ (a total function)
  is injective. \\
\spac $\vph_e$ is surjective if its restriction to~$W_e$ (a total
  function) is surjective.\\
\spac $\vph_e$ is bijective if it is total, injective and surjective.\\
Thus a ParRec functions can be injective and surjective without
being bijective.
\close
\end{definition}

There are other reasonable definitions of bijectivity for ParRec 
functions. For instance, $\vph_e$ is bijective if its restriction to~$W_e$
is bijective, as a total function.
The result stated in item~\ref{GbijP}, page~\pageref{GbijP}, is true for
both definitions.

\bigskip

\subsubsection*{Kleene Normal Form Theorem}
The Kleene Normal Form Theorem \cite{Rogers,kleene52,odi}
states that there is a PR predicate~$T$ and a PR function~$U$
such that for every ParRec function~$f(\ov{x})$ there is an index~$e$
satisfying 
$\vph_e(\ov{x})=f(\ov{x})=U(\mu_h:T(e,\ov{x},h))$, where
(i)~the predicate $T(e,\ov{x},h)$ checks if the integer~$h$ codes
  a (halting) computation history of~$f(\ov{x})$ and
(ii)~the function~$U(h)$ extracts the final result from the computation 
  history~$h$.

\begin{theorem}[A consequence of the Kleene Normal Form Theorem]
  \label{Knft}
  Let $U(e,\ov{x},t)$ be a Turing  machine that executes~$t$ steps 
  of the computation $\varphi_e(\ov{x})$, printing $\varphi_e(\ov{x})$ 
  whenever~$\varphi_e(\ov{x})$ is defined and~$t$ is sufficient large.
  Let $T(e,\ov{x},t)$ be a Turing  machine
  that executes~$t$ steps of the computation $\varphi_e(\ov{x})$, 
  returning~0 if the computation halted in $\leq t$ steps, and~1 
  otherwise (both~$U$ and $T$ always halt). 
  The number of steps needed for the computation of a halting computation 
  $\varphi_e(\ov{x})$ is thus $\mu_t:T(e,\ov{x},t)$.

  There are recursive functions~$T$ and~$U$
  such that for every ParRec function~$f$ there 
  is a integer~$e$ satisfying
  $\forall \ov{x}:[f(\ov{x}) = U(e,\ov{x},\mu_t:T(e,\ov{x},t))]$.
\end{theorem}

The function $T(e,\ov{x},t)$ mentioned in Theorem~\ref{Knft} will
often be used in this work:
\begin{definition}
  \label{interval}
  Let~$\vph_e(x)$ be the partial recursive function with index~$e$.
  Let $T(e,x,t)$ be a Turing machine corresponding to the function
  $$
  T(e,x,t) =
  \left\{
    \begin{array}{ll}
      0 & \text{if the computation of $\vph_e(x)$ converges 
                exactly at step $t$}  \\
      1 & \text{otherwise.}
    \end{array}
  \right.
  $$ 
  In most cases it does not matter if we replace ``converges exactly at step $t$''
  by ``converges at some step $\leq t$''.
  For every integer~$e$, the function $T_e(x,t)=T(e,x,t)$ is assumed to be
  primitive recursive.
  When we do not need to mention the index~$e$, we use
  the function $F(x,t)\defined T_e(x,t)$.
  Sometimes we will view the functions~$F$ and~$T$ as predicates, 
  interpreting the value~1 as ``false'' (does not halt) and the value~0 
  as ``true'' (halts).
\close
\end{definition}

\bigskip

\subsubsection*{Arithmetic hierarchy}
The arithmetic hierarchy (see for instance
\cite{odi,Rogers,kleene52,kleene55,kleene43,KP54,hinman})
is used in this paper with
the purpose of measuring the degree of undecidability of PR
decision problems. This is usual in Recursion Theory, see for instance
its use for rewriting problems studied in~\cite{enderullisETAL,simonsen}.
This is similar to the classification of the difficulty of a
(decidable) problem in the polynomial hierarchy (PH). 
As we will not mention other hierarchies, the superscript~0 in the classes 
of the AH will be dropped. For instance, we will write~$\piTwo$ instead 
of~$\Pi^0_2$.
A suitable background on the AH can be found,
for instance, in references~\cite[Chapter~7 (\S29)]{hermes},
\cite[Chapter~14]{Rogers}
\cite[Chapter~10]{cooper} , and 
\cite[Chapter~IV]{soare}. In particular, for $n\geq 1$
we have the proper inclusions $\deltaN\subset\sigN$ and
$\deltaN\subset\piN$ (Hierarchy Theorem~\cite{soare}).

We denote by~$\Delta$ the set $\sigZero=\piZero=\deltaZero=\deltaOne$. 
Thus the class of recursively enumerable, but not recursive sets
is~$\sigOne\setminus \Delta$.  

Some classes of the AH are briefly described below.
\begin{center}
  \begin{tabular}{p{30mm} p{85mm}}   
  AH class   &    Decision problem   \\ \hline
  $\Delta$           & Decidable  \\
  $\sigOne$          & Semi-decidable  \\
  $\sigOne\setminus 
           \Delta$   & Undecidable, semi-decidable  \\
  $\piOne$           & Complement of a semi-decidable problem  \\
  $\piOne\setminus 
          \Delta$    & Undecidable, complement of a semi-decidable 
                       problem     \\
  $\deltaTwo=\sigTwo\cap\piTwo$
                     & Decision problems that can be described both as
      ``$\exists x\,\forall y: P(x,y)$'' and as
      ``$\forall x\,\exists y: Q(x,y)$'' 
      where~$P$ and~$Q$ are recursive predicates ($x$ and~$y$
      represent tuples of variables).
  \end{tabular}
\end{center}
The class~$\deltaTwo$ has been studied in detail, see for instance
\cite[Section IV.1]{odi},
\cite[Chapter~XI]{odiTwo}, and~\cite{putmanH};
see also \cite{shoenfield,soskovaWu,CSY,lewisAEM}.
Another characterisation of~$\deltaTwo$ is Shoenfield's Limit 
Lemma~\cite{shoenfield}; from~\cite[page 373]{odi}:
  a set~$A$ is in~$\deltaTwo$ if and only if its characteristic function
  is the limit of a recursive function~$g$, i.e.
  $c_A(x)=\lim_{s\to\infty}g(x,s)$.

\todoi{``primitive recursive'' $\to$ ``recursive'' below.}
\begin{definition} [Many-one reduction]
  \label{many-one}
  The language~$L$
  {\em reduces} (many-one) to~$M$, $L \reduces M$, if there is a recursive
  function $f:\ene\to\ene$ such that~$x\in L$ if and only if $f(x)\in M$.
\close
\end{definition}

\begin{property}
  \label{reds}
  If $L\reduces M$ then for every integer~$n\in\ene$:
  (i)~$\ov{L}\reduces \ov{M}$,
  (ii)~$M\in\sigN$ implies $L\in\sigN$,
  (iii)~$M\in\piN$ implies $L\in\piN$. 
  (For $n=0$ we get: if~$M$ is recursive, then~$L$ is recursive.)
\end{property}

A set~$P$ is {\em m-complete} in the class~${\cal C}$ 
if~$P\in {\cal C}$ and $Q\reduces P$ for every $Q\in{\cal C}$.
We will often say ``complete'' instead of ``m-complete''.
A set~$P$ is {\em m-hard} in the class~${\cal C}$ 
if $P\in {\cal C}$ and $Q\reduces P$ for every $Q\in{\cal C}$.

{\em See also the definitions and results at the beginning of 
Section~\ref{PR-ParRec}, page~\pageref{PR-ParRec}.}
The structure of the arithmetical hierarchy classes has been widely
studied, see for instance 
\cite{post44,KP54,friedberg,shoenfield,yates1,yates2,lachlan,Rogers,
odi,odiTwo,soare}.

\subsection{Decision problems}
\label{decision}
The decision problems related to partial recursive functions will be
used for two purposes:
to prove the completeness of PR problems (using Theorem~\ref{THREEcomp} below),
and to compare the degree of undecidability of PR decision problem with 
the corresponding problem of classical Recursion Theory, see
Section~\ref{parRec}, page~\pageref{parRec}.

\todoi{Decision problem denotation}
\und{Notation.} The names of the decision problems whose 
instance is a PR function (the main subject of this paper) have 
the superscript ``PR''. 
If the instance is a ParRec function,
there is no such superscript. Examples of the former and the latter
are \eozN\ and \fdomN, respectively.\close

Definition~\ref{correspondence} (page~\pageref{correspondence}) characterises
the PR decision problem that {\em corresponds} to a given ParRec problem.

We list some important ParRec decision problems,
see~\cite{BBJ,hermes,odi,davis}.
\begin{problem}
  {\rm HP}, halting problem.\\
  \inst a Turing machine~$T$ and a word~$x$.\\
  \quest does the computation $T(x)$ halt?\closex
\end{problem}

\begin{problem}
  {\rm SHP}, self halting problem.\\
  \inst a Turing machine~$T$.\\
  \quest does the computation $T(i)$ halt, where~$i$ is a
  G\"odel index of~$T$?\closex
\end{problem}

\begin{problem}
  {\rm PCP}, Post correspondence problem.\\
  \inst two finite sets of non-null words, $\{x_1,\ldots,x_n\}$ and
  $\{y_1,\ldots,y_n\}$, defined over the same
  finite alphabet with at least two letters.\\
  \quest is there a finite non-null sequence
  $i_1,\ldots,\, i_k$ such that
  $x_{i_1}x_{i_2}\ldots x_{i_k} = y_{i_1}y_{i_2}\ldots y_{i_k}$?\closex
\end{problem}

\todox{Proof deleted, see Theorem~\ref{THREEcomp},
  page~\pageref{THREEcomp}.}
\begin{problem}
  {\totalN}.\\
  \inst a ParRec function~$f(x)$.\\
  \quest Is~$f(x)$ defined for every~$x$?\closex
\end{problem}

\todoi{New partial recursive problems included\ldots}
\begin{problem}
  {\fdomN}.\\
  \inst a ParRec function~$f(x)$.\\
  \quest is the domain of~$f(x)$ finite?\closex
\end{problem}
 
\begin{problem}
  {\cfdomN}.\\
  \inst a ParRec function~$f(x)$.\\
  \quest is the domain of~$f(x)$ cofinite?\closex
\end{problem}
 
\begin{problem}
  {\GequiN}.\\
  \inst ParRec functions~$f(x)$ and~$g(x)$.\\
  \quest Is $f=g$ (same function)?\closex
\end{problem}
 
\begin{theorem}
  \label{THREEcomp} {\ } \\
  \fdomN\  is $\sigTwo$-complete.
  \totalN\ is $\piTwo$-complete.
  \cfdomN\ is $\sigThree$-complete.\\
  \GequiN\ is $\piTwo$-complete.
\end{theorem}

\begin{proof}
  For the problems \fdomN, \totalN, and \cfdomN\ see for 
  instance~\cite{soare}.\\
  Let~$e$ and~$e'$ be the instance of \GequiN. The equality of the 
  functions~$\vph_e$ and~$\vph_{e'}$ can be expressed as follows
  $$
  \begin{array}{ll}
    \forall x:& 
    \overbrace{(\forall t_1:T(e,x,t_1)=T(e,x,t_1)=1)}^{
      \text{both functions undefined}}  \;\;\vee  \\ \\
    &
    \overbrace{\exists t,t': (T(e,x,t)=T(e,x,t')=0) \wedge
               U(e,x,t)=U(e',x,t')}^{
      \text{both functions defined and with the same value}}
  \end{array}
  $$
  (see Definition~\ref{Knft}, page~\pageref{Knft}).

  This can be expressed in the form 
  $\forall x,t_1\:\exists t,t':\ldots$, so that the problem
  \GequiN\ is in $\piTwo$.

  To prove completeness
  we reduce the $\piTwo$-complete problem \totalN\ 
  to \GequiN.
  Given the ParRec function~$f(x)$, define the ParRec 
  function~$g(x)$ as:\\
  \spac 1) compute $f(x)$\\
  \spac 2) if the computation halts, output 0 
            (otherwise the result is undefined).

  Clearly~$g(x)$ is equivalent to the zero function~$0(x)$
  iff~$f(x)$ is total.
\end{proof}

\todox{Removed a note about assumptions in~\cite{soare} (problem \fdomN.)}

\subsection{Models of computation}
\label{models}
Deterministic finite automata (DFA), Turing machines (TM), and context
free grammars (\CFG) and examples of ``models of computation''.
A fundamental assumption associated with the
intuitive idea of ``model of computation'' is the following.
\begin{assumption}
\label{model}
The instances of a ``model of computation'' are
recursively enumerable.
\end{assumption}
For instance, we may  conclude that no model of computation
characterises the set of recursive (total) functions because that set is
not enumerable (see for instance~\cite{Phillips}).

\subsubsection*{The ``dovetailing'' technique}
As an illustration of this technique consider a recursively enumerable 
set of Turing machines~$M_1$,
$M_2$\ldots\ and suppose that we want to list the outputs of all the
computations~$M_i(x)$ for some fixed~$x$. That is, for every~$i$, if
the computation~$M_i(x)$ halts, its output is listed. This can be effectively
done by the following method:

-- Successively simulate one step of computation for each machine, in
the following order
{\footnotesize
$$
    \begin{array}{rccccccccccccccccc}
      \text{M:}&&
        M_1 && M_1&M_2 && M_1&M_2&M_3 && M_1&M_2&M_3&M_4 && M_1&\ldots\\
      \text{MS:}&&
        1   && 2  & 1  && 3  & 2 & 1  &&  4 & 3 & 2 & 1  &&  5\\
      \text{GS:}&&
        1   && 2  & 3  && 4  & 5 & 6  &&  7 & 8 & 9 & 10 && 11
    \end{array}
$$
}
  where~M, MS, and GS denote respectively ``machine'', 
  ``machine step number'' and ``global step number''.
  During this simulation, whenever some machine~$M_i$ halts, 
  print its output.  
  Note that every every computation step of every machine is simulated: 
  either forever or until it halts. 

  More generally we can for instance use the dovetailing technique to
  enumerate any r.e. collection of r.e. sets.

\bigskip

\subsection{Primitive recursive functions}
\label{pr-funs}
It is assumed that the reader is familiar with:
\begin{enumerate}
\item [--] the classical definition of primitive recursive,
  functions, see~\cite{BBJ,odi};
\item [--] the use of register languages to characterise those 
  functions (particularly the Loop language, 
  see~\cite{MeyerRitchie67a,MeyerRitchie67b});
\item [--] the PR binary function~``$\dminus$''
  $$
   x\dminus y =
   \left\{
     \begin{array}{ll}
       x-y & \text{\spac if $x\geq y$}  \\
       0   & \text{\spac otherwise.}
     \end{array}
   \right.
  $$
\end{enumerate}

The following property will be useful.
\begin{property}
  \label{eq-funs}
  Let~$a$ and~$b$ be arbitrary integers. Then
  $a = b$ if and only if
  $(a\dminus b) + (b\dminus a)=0$.
\end{property}

\todox{Check. The reviewer found \Smn\ Theorem ``non standard''.}
\todox{CFG's not needed now. Deleted.}
\section{Decision problems associated with PR functions}
\label{problems}
In this section we study the decidability of some 
problems\footnote{Many other natural decision problems could have been studied.
For space reasons we did not include (i)~several simple generalisations and 
(ii)~similar problems (in terms of the degree of undecidability).
An example of~(i) is ``does~$f(x)$ have a number
of zeros between~$m$ and~$n$?'' (for fixed~$m$ and~$n$ with $0<m<n$); this problem
generalises problems~\ref{eoz} and~\ref{ekz}.
An example of~(ii) is the following problem which is similar to \recurN\ 
(Problem~\ref{Precur}, page~\pageref{Precur}):
``given~$e$, $\exists e':\text{$e'\in W_e$ and $W_{e'}$ is infinite}$?''
(see \cite[\S 14.8, Theorem~XV (page 326)]{Rogers}).
} for which the instance is one PR function (or a pair of PR functions).
First, in Section~\ref{various}, we study several problems related to
the existence of zeros of the PR function~$f$ and to the equivalence of
the PR functions~$f$ and~$g$, then, in Section~\ref{codomain}, some
properties related to the codomain of~$f$ are studied, and finally, in
Section~\ref{injbij} we prove results associated to the injectivity,
surjectivity, and bijectivity of~$f$. In every case it is assumed that
the PR function is given by an index (or equivalently by a Loop program).
The reader can find a summary of some of these results in Figure~\ref{summary}
(page~\pageref{summary}). 

In Section~\ref{compare} (page~\pageref{compare}) we also consider the
case in which the instance is a ParRec function and compare
the two cases. Figures~\ref{c-summary} (page~\pageref{c-summary})
and~\ref{changes} (page~\pageref{changes}) contain a summary of this
comparison.

\subsection{Some undecidable problems}
\label{various}
\subsubsection{Problems}
For each of the decision problems listed below
\label{variousP}
the instance consists of one or two primitive recursive functions.
When the instance has the form $f(\ov{x},y)$, where~$\ov{x}$ is given,
we can use the \Smn\ Theorem 
to convert it to the form~$g(y)$, where~$g$ depends on~$\ov{x}$. 
Both these forms will be used in the sequel.
All the problems are all proved to be either decidable or
$\sigN$-complete (for some~$n\geq 1$) or
$\piN$-complete (for some~$n\geq 1$) or members of the class
$\deltaTwo\setminus(\sigOne\cup \piOne)$.

\todoi{All the problems in the list: proved complete OR in 
    $\deltaTwo\setminus(\sigOne\cup \piOne)$.}

\begin{problem}
  {\ozN}.
  \label{oz}\\
    \inst $\ov{x}$ and the primitive recursive function $f(\ov{x},y)$.\\
    \quest Is there some~$y$ such that $f(\ov{x},y)=0$? \\
    \class $\sigOne$-complete, Theorem~\ref{t-oz} (page~\pageref{t-oz}).\closex
\end{problem}

\begin{problem}
  {\eozN}.
  \label{eoz}  \\
    \inst $\ov{x}$ and the primitive recursive function $f(\ov{x},y)$.\\
    \quest Does $f(\ov{x},y)$ have exactly one zero? \\
    \class $\deltaTwo\setminus  (\sigOne\cup \piOne)$, 
    Theorem~\ref{t-eoz} (page~\pageref{t-eoz}).\closex
\end{problem}

\begin{problem}
  {\gtkzN}. 
  \label{gtkz} \\
    \inst $\ov{x}$ and the primitive recursive function $f(\ov{x},y)$.\\
    \quest Does $f(\ov{x},y)$ have at least $k$ zeros? 
               ($k\geq 1$ fixed). \\
    \class $\sigOne$-complete, 
    Theorem~\ref{t-gtkz} (page~\pageref{t-gtkz}).\closex
\end{problem}

\begin{problem}
  {\ekzN}.
  \label{ekz} \\
    \inst $\ov{x}$ and the primitive recursive function $f(\ov{x},y)$.\\
    \quest Does $f(\ov{x},y)$ have exactly $k$ zeros?
               ($k\geq 2$ fixed). \\
    \class $\deltaTwo\setminus (\sigOne\cup\piOne)$, 
    Theorem~\ref{t-ekz} (page~\pageref{t-ekz}).\closex
\end{problem}

\begin{problem}
  {\eqNXTN}.
    \label{eqNXT} \\
    \inst $\ov{x}$ and the primitive recursive function $f(\ov{x},y)$.\\
    \quest Is there some~$y$ such that $f(\ov{x},y)=f(\ov{x},y+1)$?\\
     \class $\sigOne$-complete, Theorem~\ref{t-eqNXT}
    (page~\pageref{t-eqNXT}).\closex
\end{problem}

\begin{problem}
  {\azN}.
  \label{az} \\
    \inst $\ov{x}$ and the primitive recursive function $f(\ov{x},y)$.\\
    \quest Is $f(\ov{x},y)$ the zero function~$0(y)$?\\
    \class $\piOne$-complete, Theorem~\ref{t-az} (page~\pageref{t-az}).\closex
\end{problem}

\begin{problem}
  {\infzN}.
  \label{infz} \\
    \inst $\ov{x}$ and the primitive recursive function $f(\ov{x},y)$.\\
    \quest Does $f(\ov{x},y)$ have infinitely many zeros? \\
    \class $\piTwo$-complete.
    Theorem~\ref{t-infzCOMP} (page~\pageref{t-infzCOMP}).\closex
\end{problem}

\todoi{New PR problem.}
\begin{problem}
    {\aazN}.
    \label{aaz} \\
    \inst $\ov{x}$ and the primitive recursive function $f(\ov{x},y)$.\\
    \quest Is $f(\ov{x},y)=0$ for almost all values of~$y$? \\
    \class $\sigTwo$-complete,
    Theorem~\ref{t-aaz} (page~\pageref{t-aaz}).\closex
\end{problem}

\begin{problem}
  \label{DeOne}
  \eOneN, ``functions equal at one point''.         \\
  \inst $\ov{x}$ and the primitive recursive functions $f(\ov{x},y)$, 
        $g(\ov{x},y)$. \\
  \quest Is $f(\ov{x},y)=g(\ov{x},y)$ for at least one~$y$?\\
  \class $\sigOne$-complete, Theorem~\ref{eq-one} 
  (page~\pageref{eq-one}).\closex
\end{problem}

\begin{problem}
\label{equiP}
  {\equiN}.    \\
  \inst $\ov{x}$ and the primitive recursive functions $f(\ov{x},y)$, 
        $g(\ov{x},y)$. \\
  \quest Is $f(\ov{x},y)=g(\ov{x},y)$ for every~$y$? \\
  \class $\piOne$-complete, Theorem~\ref{t-equi} (page~\pageref{t-equi}).\closex
\end{problem}

\subsubsection*{A note on polynomials and the \ozN\ problems}
For each instance~$f(x)$ of the \ozN\ problem there is
a polynomial with integer coefficients $p(x_1,\ldots, x_n)$,
such that~$p$ has a zero iff~$f(x)$ has a zero.
Let us begin by quoting~\cite{feferman}.
\begin{quote} 
\noindent {In an unpublished and undated manuscript from the 1930s found in 
  G\"odel’s Nachlass and reproduced in Vol. III of the Collected Works,
  he showed that every statement of the form~``$\forall x:R(x)$'' 
  with~$R$ primitive recursive [relation] is equivalent to one in 
  the form
  $$
     \forall x_1,\ldots, x_n \; \exists y_1,\ldots, y_m :
       [p(x_1,\ldots x_n,y_1\ldots y_m)=0]
  $$
  in which the variables range over natural numbers and~$p$ is a
  polynomial with integer coefficients; it is such problems that 
  G\"odel referred to as Diophantine in the Gibbs lecture. It follows
  from the later work on Hilbert’s 10th problem by Martin Davis,
  Hilary Putnam, Julia Robinson and --- in the end --- 
  Yuri Matiyasevich that, even better, one can take~$m=0$ in such a
  representation, when the~``$=$'' relation is replaced by~``$\neq$''.}
\end{quote}
In terms of PR functions, a PR relation~$R(x)$ may be represented by
a PR function~$f(x)$ such that~$R(x)$ holds iff~$f(x)\neq 0$.
We have (with~$m=0$) the equivalences
$$
\begin{array}{lcl}
  \forall x:R(x)      &\lra& 
        \forall x_1,\ldots, x_n : [p(x_1,\ldots x_n)\neq 0]\\
  \forall x:f(x)\neq 0&\lra& 
        \forall x_1,\ldots, x_n : [p(x_1,\ldots x_n)\neq 0]\\
  \exists x:f(x)=0    &\lra& 
        \exists x_1,\ldots, x_n : [p(x_1,\ldots x_n)=0]\Mdot
\end{array}
$$
Thus, given that there is an algorithm that computes~$p$ from~$f$,
and relatively to the decidability of the \ozN\ problem, an
arbitrary PR function is no more general than a (multiple variable) integer
polynomial.

\subsubsection{Results}
We use the following easy to prove fact: for
every~$n\geq 1$ if the language~$P$ is complete in the class~$\sigN$,
then the language~$\neg P$ is complete in the class~$\piN$; and if the
language~$P$ is complete in the class~$\piN$, the language~$\neg P$ is
complete in the class~$\sigN$.

\todoi{Problem complete in $\sigOne$.}
\begin{theorem}
  \label{t-oz}
  The \ozN\ problem is $\sigOne$-complete.
\end{theorem}

\begin{proof}
  The semi-decidability is obvious. 
  To prove the completeness, we reduce the $\sigOne$-complete problem SHP
  to \ozN.
  Let~$e$ be an instance of SHP. The computation $\varphi_e(e)$ halts
  iff there is some~$t$ such that $T(e,e,t)=0$ (see Definition~\ref{interval},
  page~\pageref{interval}).
  That is, $\varphi_e(e)\halts\:\lra\:(\exists t:T_e(t)=0)$,
  where~$T_e$ is an unary Turing machine.
  Moreover, $T_e(t)$ is primitive recursive.  
  Thus, $\SHP\reduces\oz$ and we conclude that \ozN\ is 
  $\sigOne$-complete.
\end{proof}

\todoi{Problem in $\deltaTwo \setminus (\sigOne\cup\piOne)$.}
\begin{theorem}
  \label{t-eoz}
  The language associated with the \eozN\ problem is in
  $\deltaTwo \setminus (\sigOne\cup\piOne)$.
\end{theorem}

\begin{proof}
  (i)~In~$\deltaTwo$.
  Given the PR function~$f(x)$, the
  question associated with the problem can be expressed as
  \begin{equation}
    \label{eozEXP}
     (\exists x:\overbrace{f(x)=0}^{A}) \:\wedge\:
     (\forall x_1,x_2:
        \overbrace{(x_1=x_2)\vee f(x_1)\neq 0\vee f(x_2)\neq 0}^{B})\Mdot
  \end{equation}
  As the variable~$x$ does not occur in~$B$ and neither~$x_1$ nor~$x_2$ 
  occur in~$A$, we can write~\eqnref{eozEXP} in the following two forms
  \begin{eqnarray}
    &&
    \label{formA}
     \exists x\:
     \forall x_1,x_2:
        (f(x)=0) \:\wedge\:
        ({(x_1=x_2)\vee f(x_1)\neq 0\vee f(x_2)\neq 0})  \\
    &&
    \label{formB}
     \forall x_1,x_2\:
     \exists x:
        (f(x)=0) \:\wedge\:
        ((x_1=x_2)\vee f(x_1)\neq 0\vee f(x_2)\neq 0)\Mdot
  \end{eqnarray}
  The representations~\eqnref{formA} and~\eqnref{formB} show that~\eozN\ 
  belongs both to~$\sigTwo\cap\piTwo=\deltaTwo$.

  \noindent(ii)~Not in~$\piOne$.
  Reduce \ozN\ to \eozN.
  Let~$f$ be the instance of \ozN. Define the PR function~$g$ as
  $$
  g(x) =
  \left\{
    \begin{array}{ll}
      0  & \text{\spac if $f(x)=0$ and $[x'<x \Rightarrow f(x')\neq 0]$}\\
      1  & \text{\spac otherwise}
    \end{array}
  \right.
  $$
  Clearly~$f$ has at least one zero iff the PR function~$g$ has exactly
  one zero. Thus the problem \eozN\ is not in~$\piOne$.

  \noindent(iii)~Not in~$\sigOne$.
  Reduce $\neg$(\ozN) to \eozN.
  Let~$f$ be the instance of $\neg$(\ozN). Define~$g$ as
    $$
    g(x) =
    \left\{
    \begin{array}{ll}
      0                & \text{\spac if $\;x=0$}                      \\
      1                & \text{\spac if $x\geq 1$ and $\;f(x-1)\neq 0$}             \\
      0                & \text{\spac if $x\geq 1$ and $\;f(x-1)= 0$}
    \end{array}
    \right.
    $$
  The function~$g$ has exactly one zero iff the function~$f$ has no
  zeros. It follows that \eozN\ is not in~$\sigOne$.
\end{proof}

\todoi{Problem complete in $\sigOne$.}
\begin{theorem}
  \label{t-gtkz}
  The \gtkzN\ problem is $\sigOne$-complete.
\end{theorem}

\begin{proof}
  The \gtkzN\ problem is clearly semi-decidable and \ozN\
  (instance~$f(x)$) reduces easily to \gtkzN\ (instance~$g(x)$) if we
  define for each~$x\geq 0$
  $$
    g(kx) = g(kx+1) = \ldots = g(kx+(k-1)) = f(x)
  $$
  so that each zero of~$f$ corresponds to~$k$ zeros of~$g$.
  As the problem \ozN\ is complete in~$\sigOne$
  (Theorem~\ref{t-oz}, page~\pageref{t-oz}),
  the reduction from \ozN\ to \gtkzN\ proves the $\sigOne$-completeness 
  of \gtkzN.
\end{proof}

\todoi{Problem  in $\deltaTwo \setminus (\sigOne\cup\piOne)$.}
\begin{theorem}
  For every~$k\geq 1$ the language associated with the \ekzN\ problem
  is in $\deltaTwo \setminus (\sigOne\cup\piOne)$.
  \label{t-ekz}
\end{theorem}

\begin{proof}
  Similar to the proof of Theorem~\ref{t-eoz}.
\end{proof}

\todoi{Problem complete in $\sigOne$.}
\begin{theorem}
  \label{t-eqNXT}
  The \eqNXTN\ problem is $\sigOne$-complete.
\end{theorem}

\begin{proof}
  It is obvious that $\eqNXT\in \sigOne$.\\
  The following observation suggests a reduction of \ozN\ (instance~$f$)
  to \eqNXTN\ (instance~$g$): $f(x)$ has at least one zero
  iff $g(x)\defined \sum_{i=0}^{i<x} f(i)$ has at least a value
  equal to the next one (note that $g(0)=0$).
  As \ozN\ was proved to be $\sigOne$-complete (Theorem~\ref{t-oz},
  page~\pageref{t-oz}), this reduction
  shows that \eqNXTN\ is also $\sigOne$-complete.
\end{proof}

\todoi{Problem complete in $\piOne$.}
\begin{theorem}
  \label{t-az}
  The \azN\ problem is $\piOne$-complete.
\end{theorem}

\begin{proof}
  The problem can be expressed as $\forall x:f(x)=0$.
  Thus it belongs to~$\piOne$.   \\
  Reduce $\oz$ to $\neg(\az)$.  Given an instance~$f$ of \ozN, 
  define the function~$g$ as
  $$
  g(x)=
  \left\{
    \begin{array}{ll}
      0  &  \text{\spac if $f(x)\neq 0$} \\
      1  &  \text{\spac if $f(x)=    0$}
    \end{array}
  \right.
  $$
  The function~$g$, which is clearly PR, has at least one zero 
  iff~$f$ is {\em not} the zero function.
  As \ozN\ is $\sigOne$-complete (Theorem~\ref{t-oz} page~\pageref{t-oz}),
  this reduction proves the theorem.
\end{proof}

\todoi{Problem complete in $\piTwo$.}
\begin{theorem}
\label{t-infzCOMP}
  The \infzN\ problem is is $\piTwo$-complete.
\end{theorem}

\begin{proof}
  In~$\piTwo$:
  the statement associated with the problem can be expressed as
  $$
    \forall m\,\exists\,x: (x\geq m) \wedge (f(x)=0)\Mdot
  $$

%
%
  To prove completeness,
  consider~$P$, the complement of the problem \infzN\ (``finite number of 
  zeros'').
  We prove the completeness of~$P$ in~$\sigTwo$ using a
  reduction of the $\sigTwo$-complete problem \fdomN\ 
  (Theorem~\ref{THREEcomp}, page~\pageref{THREEcomp}) to~$P$.
  Consider the function $T_e(\ang{x,t})\defined T(e,x,t)$ 
  (Definition~\ref{interval}, page~\pageref{interval}).
  We assume that for each~$x$ there is at most one~$t$ such that $T(e,x,t)=0$;
  thus, the number of zeros of~$T_e$ equals the size of the domain of~$\vph_e$.
  The instance of the class~$P$ that corresponds to the instance~$\vph_e$
  of \fdomN\ is defined as~$T_e$.
  Clearly $T_e$ has a finite number of zeros iff~$\vph_e$ is in \fdomN.
\end{proof}

\todoi{New problem complete in $\sigTwo$.}
\begin{theorem}
  \label{t-aaz}
  The \aazN\ problem is is $\sigTwo$-complete.
\end{theorem}

\begin{proof}
  First notice that \aazN\ is in~$\sigTwo$, because a function~$f$ 
  is in \\
  \aazN\ iff
  $$
  \exists x_0\forall x: x\geq x_0 \imply f(x)=0\Mdot
  $$
  To prove completeness,
  use Theorem~\ref{THREEcomp} (page~\pageref{THREEcomp}) and the 
  following characterisation of \\
  \fdomN
  $$
   \exists a\forall x\forall t: 
   (x\geq a \wedge t\geq a) \imply 
   \text{(the computation $T(e,x,t)$ did not halt in time~$\leq t$)}
  $$
  to define a reduction of \fdomN\ to the PR problem \aazN: 
  given a ParRec function~$\vph_e$ (instance of \fdomN), consider the PR
  function $T'_e(\ang{x,t})\defined 1-T(e,x,t)$ 
  (see Definition~\ref{interval}, page~\pageref{interval};
  assume that for each~$x$ there is at most one~$t$ such that $T(e,x,t)=0$).
  The function~$T'_e$ is in the class \aazN\ iff~$\vph_e$ is in \fdomN.
\end{proof}

\todoi{Problem complete in $\sigOne$.}
\begin{theorem}
  \label{eq-one}
  The \eOneN\ problem is $\sigOne$-complete.
\end{theorem}

\begin{proof}
  Recall Definition~\ref{DeOne}, page~\pageref{DeOne}.
  The problem \eOneN\ is clearly in~$\sigOne$.
  The \ozN\ easily reduces to this problem if we fix $g(\ov{x},y)=0$
  (the zero function).
  The completeness follows from this reduction and the fact that 
  \ozN\ is $\sigOne$-complete,  Theorem~\ref{t-oz} (page~\pageref{t-oz}).
\end{proof}

\todoi{Problem complete in $\piOne$.}
\begin{theorem}
  \label{t-equi}
  The \equiN\ problem is $\piOne$-complete.
\end{theorem}

\begin{proof}
  The complement of the \equiN\ problem is clearly semi-decidable. \\
  Consider the instance $\pair{e,x}$ of HP. 
  Using the Kleene Normal Form, we see that the computation~$\vph_e(x)$
  halts iff the corresponding PR function $T_{e,x}(t)\defined T(e,x,t)$ 
  (see Theorem~\ref{Knft}, page~\pageref{Knft})
  has one zero. 
  The PR function
  $$
  h(t) =
  \left\{
  \begin{array}{ll}
    0 & \text{\spac if $T_{e,x}(t)\neq 0$}  \\
    1 & \text{\spac if $T_{e,x}(t)= 0$}
  \end{array}
  \right.
  $$
  {\em is not} equivalent to the zero function~$0(t)$ iff
  $T_{e,x}(t)$ has at least one zero. This defines
  a reduction $\neg\HP$ to \equiN\ which proves the completeness.
\end{proof}

\subsection{Size of the codomain}
\label{codomain}
We now study decision problems related to the size of the codomain 
of a PR function. The \infCodN\ problem
will be used in Sections~\ref{Pfh} (page~\pageref{Pfh}) 
and~\ref{Pgfh} (page~\pageref{Pgfh}).
In the following problems, the instance is the PR function~$f$.

\begin{problem}
\kCodN, the codomain is finite with cardinality~$k$.\\
\quest Does the codomain of~$f$ have size~$k$? \\
\class $\piOne$-complete for~$k=1$ (Theorem~\ref{kCodT},
page~\pageref{kCodT}). \\
\class $\deltaTwo\setminus (\sigOne\cup\piOne)$
for~$k\geq 2$ (Theorem~\ref{kCodTwo}, page~\pageref{kCodTwo}).\closex
\end{problem}

\todoi{Problem complete in $\piOne$.}
\begin{theorem}
  \label{kCodT}
  The language associated with the problem \kCodON\ is $\piOne$-complete.
\end{theorem}

\begin{proof}
  A positive answer to (\kCodON) can be expressed as
  $\forall x: (f(x)=f(0))$, so that (\kCodON) is in~$\piOne$.
  Reduce (\azN) to (\kCodON), as follows. Given an instance~$f$ of
  \azN\ define the instance~$g$ of (\kCodON) as
  $$
  \left\{
    \begin{array}{ll}
      g(0)=0  &  \\
      g(x)=f(x-1)  &  \text{\spac for $x\geq 1$}\Mdot
    \end{array}
  \right.
$$
The function~$g$ has codomain with size~1 iff~$f(x)=0$ (zero function).
As \azN\ is $\piOne$-complete (Theorem~\ref{t-az}, page~\pageref{t-az}), 
\kCodON\ is also $\piOne$-complete.
\end{proof}

\todoi{Problem  in $\deltaTwo \setminus (\sigOne\cup\piOne)$.}
\begin{theorem}
  \label{kCodTwo}
  For any integer~$k\geq 2$ the language associated with the
  problem \kCodN\ belongs to the class 
  $\deltaTwo\setminus (\sigOne\cup\piOne)$.
\end{theorem}

\begin{proof} 
Assume $k\geq 2$.\\
\noindent(i)
In~$\deltaTwo$:
   a positive answer to \kCodN\ can be expressed as
   (illustrated for the case~$k=2$)
   $$
     [\exists x_1,x_2:A(x_1,x_2)] \:\wedge\:
     [\forall z_1,z_2,z_3:B(z_1,z_2,z_3)]
   $$
   where 
   $$
   \begin{array}{lcll}
     A(x_1,x_2)    &=& f(x_1)\neq f(x_2) &
     \spac\text{$\abs{\cod{f}}\geq 2$} \\
     B(z_1,z_2,z_3)&=& 
        (f(z_1)=f(z_2))\vee(f(z_2)=f(z_3))\vee(f(z_3)=f(z_1)) &
     \spac\text{$\abs{\cod{f}}<3$}
   \end{array}
   $$
   The question associated with the problem can thus be expressed
   in~2 forms:
  \begin{eqnarray*}
    \label{formX}
     \exists x_1,x_2\:
     \forall z_1,z_2,z_3:
        A(x_1,x_2) \wedge B(z_1,z_2,z_3) \\
    \label{formY}
     \forall z_1,z_2,z_3\:
     \exists x_1,x_2:
        A(x_1,x_2) \wedge B(z_1,z_2,z_3)
  \end{eqnarray*}
  Thus \kCodN\ belongs both to $\deltaTwo=\sigTwo\cap\piTwo$.

\noindent(ii)
   Not in~$\sigOne$:
   reduce \azN\ to \kCodN: 
   let~$f$ be an instance of
   \azN. Define~$g$ as
   $$
   g(x) = \left\{
     \begin{array}{ll}
       x        & \text{\spac for $0\leq x< k$}\\
       k\times f(x-k) & \text{\spac for $x\geq k$}
     \end{array}
   \right.
   $$
   Clearly $\abs{\cod{g}}=k$ iff~$f(x)=0$ (zero function).
   Thus \kCodN\ is not in~$\sigOne$.

\noindent(iii)
   Not in~$\piOne$:
   reduce ($\neg$\azN) to \kCodN:
   let~$f$ be an
   instance of ($\neg$\azN). Define
   $$
   g(kx+i) = \left\{
     \begin{array}{lll}
       0 & \text{\spac if $f(x)=0$}      
           &\spac\text{(for $i=0$, $1$,\ldots, $k-1$)}\\
       i & \text{\spac if $f(x)\neq 0$}  
           &\spac\text{(for $i=0$, $1$,\ldots, $k-1$)}\Mdot
     \end{array}
   \right.
   $$
   Clearly $|\cod{g}|=k$ iff~$f(x)\neq 0$ for at least a value of~$x$.
   Thus \kCodN\ is not in~$\piOne$.
 \end{proof}

\todoi{Problem complete in $\sigTwo$.}

\begin{problem}
\fCodN, finite codomain.    \\
\quest Is the codomain of~$f$ finite?       \\
\class $\sigTwo$-complete (Theorem,~\ref{fCodT}, page~\pageref{fCodT}).\closex
\end{problem}

\begin{theorem}
\label{fCodT}
  The problem \fCodN\ is $\sigTwo$-complete.
\end{theorem}

\begin{proof}
  \indent In~$\sigTwo$: the \fCodN\ statement can be expressed as
  $\exists m \forall x: f(x)\leq m$.

%

\noindent
  $\sigTwo$-complete:
  the problem \infzN\ is $\piTwo$-complete (Theorem~\ref{t-infzCOMP},
  page~\pageref{t-infzCOMP}). Reduce \\
  $\neg$(\infzN) to \fCodN.
  The PR function~$f$ has a finite number of zeros iff the PR function defined by
  the following program has finite codomain.

   \bigskip

  \begin{minipage}{0.75\linewidth}
  \noindent\spacc
  {\tt Function} $g(n)$: \\
  \spacc $m\assign 0$;\\
  \spacc {\tt for} $i=0$, 1,\ldots, $n$:\\
  \spacc\spac {\tt compute} $f(i)$;\\
  \spacc\spac {\tt if} $f(i)=0$:\\
  \spacc\spac\spac $m\assign m+1$;\\
  \spacc {\tt return} $m$;
  \end{minipage}\\
\ 
\end{proof}

The problem \infCodN\ is the negation of the problem 
\fCodN. It follows
from Theorem~\ref{fCodT} that \infCodN\ is $\piTwo$-complete.

\color{black}
\subsection{Injectivity, surjectivity, and bijectivity}
\label{injbij}
The problems of deciding if a given PR function is injective,
surjective, or bijective are considered in this section. 

\begin{problem}
  \und{\prinjN}, primitive recursive injectivity.\\
  \inst a PR function~$f$.\\
  \quest is~$f$ an injective function?\closex
\end{problem}

\begin{problem}
  \und{\prsurjN}, primitive recursive surjectivity.\\
  \inst a PR function~$f$.\\
  \quest is~$f$ a surjective function?\closex
\end{problem}

\begin{problem}
  \und{\prbijN}, primitive recursive bijectivity.\\
  \inst a PR function~$f$.\\
  \quest is~$f$ a bijective function?\closex
\end{problem}

\subsubsection{Injectivity}
\todox{Lemma on CFG's not needed. Simpler proof below.}


\todoi{Problem complete in $\piOne$.}
\begin{theorem}
  \label{Tprinj}
  The problem \prinjN\ is $\piOne$-complete.
\end{theorem}
As a consequence of this result the injective PR functions can not be
effectively enumerated and can not be characterised by a ``model of
computation'', see Assumption~\ref{model} (page~\pageref{model}).

\begin{proof}
  The \prinjN\ statement can be expressed as
  $$
  \forall m, n : (m\neq n) \Rarr (f(m)\neq f(n))
  $$
  Thus \prinjN\ belongs to the class~$\piOne$.\\
  The \ozN\ is $\sigOne$-complete (Theorem~\ref{t-oz}, 
  page~\pageref{t-oz}). We reduce \ozN\ to \\
  $\neg$\prinjN.  
  Let~$f$ be the instance of \ozN. Define the function~$g$ as
  $$
  \left\{
  \begin{array}{ll}
    g(0)=0           \\
    g(n)=n   & \text{\spac if $n\geq 1$ and $f(n-1)\neq 0$}\\
    g(n)=0   & \text{\spac if $n\geq 1$ and $f(n-1)= 0$}
  \end{array}
  \right.
  $$
  Clearly, $g$ is injective iff~$f$ {\em has no} zeros.
\end{proof}

\color{black}
\subsubsection{Surjectivity and bijectivity}
\label{PRsur}





\todox{This result supersedes previous results (deleted).}
\todoi{Problem complete in $\piTwo$.}
\begin{theorem}
  \label{tt-surj}
  The problem \prsurjN\ is $\piTwo$-complete.
\end{theorem}

\begin{proof}
  An instance~$f$ of \prsurjN\ can be expressed as 
  $\forall y\exists x:f(x)=y$.
  It follows that \prsurjN\ is in the class~$\piTwo$.

  To prove completeness, we reduce the $\piTwo$-complete problem
  $\neg\fdom$ (Theorem~\ref{THREEcomp}, page~\pageref{THREEcomp})
  to \prsurjN. Let~$\vph_e$ be an instance of~$\neg\fdom$.
  Consider the Turing machine~$T$ in Definition~\ref{interval} 
  (page~\pageref{interval}). 
  Define the PR function~$f$ as follows:
  $f(n)$ is number of integers $m<n$ with\footnote{By ``$m=\ang{x,t}$''
  we mean: use $\ang{\cdot,\cdot}$, the standard bijection 
  $\ene^2\to\ene$, to extract~$x$ and~$t$ from~$m$.} 
  $m=\ang{x,t}$ for which $T(e,x,t)=0$, that is, for which the computation 
  $\vph_e(x)$ halts at exactly the step~$t$. 
  Clearly the codomain of the PR function~$f$ is~$\ene$ iff 
  the ParRec function~$\vph_e$ {\em does not} have finite 
  domain.
\end{proof}

\todox{This result supersedes previous results (deleted).}
\todoi{Problem complete in $\piTwo$.}
\begin{theorem}
  \label{tt-bij}
  The problem \prbijN\ is $\piTwo$-complete.
\end{theorem}

\begin{proof}
  An instance~$f$ of \prbijN\ can be expressed as 
  $$
  \forall y, x_1, x_2\,\exists x: [f(x)=y]
          \wedge
          [x_1\neq x_2 \;\Rightarrow\; f(x_1)\neq f(x_2)]\Mdot
  $$
  Thus \prbijN\ belongs to the class~$\piTwo$.

  We reduce the $\piTwo$-complete problem \prsurjN\ (Theorem~\ref{tt-surj} 
  above) to \prbijN.
  Let~$f$ be the instance of \prsurjN. For convenience we will use
  the function $g(n)\defined f(\lfloor n/2\rfloor)$ so that there 
  are infinitely many pairs $(m,n)$ with $g(m)=g(n)$.

  The function~$h$, instance of \prbijN, is defined as 

  \bigskip
  \noindent\spacc
  \begin{minipage}{0.92\linewidth}
  {\tt Function} $h(n)$: \\
  {Compute the values} $g(0)$,\ldots, $g(n)$;\\
  {if} $g(n)\not\in\{g(0),\ldots,g(n-1)\}$ {\tt then}: \\
    \spac $h(n)=2\times g(n)$;  
       \spac\hspace*{1mm}{\tt // (if $g(n)$ is a new value)}\\
  {\tt else}: \\
    \spac $h(n)=${\tt first unused odd integer}  \\
        \spac{\tt // (if $g(n)=g(i)$ for some $i<n$)}
  \end{minipage}
  \bigskip

  The codomain of~$h$ consists of two parts: (i)~the set of even integers,
  having the form $2\times a$ where~$a$ belongs to 
  the codomain of~$f$,
  and (ii)~the set of all the odd integers, $\{1,3,5,\ldots\}$.
  Clearly~$h$ is always injective. 
  If~$f\in\prsurj$, all the even integers occur in the set~(i) 
  so that~$h\in\prbij$.
  If~$f\not\in\prsurj$, some integer~$a$ is not in the codomain of~$f(i)$.
  Thus~$2a$ does not belong to the set~(i),
  so that~$h\not\in\prbij$.
  Moreover the function~$h$ is PR (if~$f$ is PR).

  The table below illustrates the definition of the functions~$g$
  and ~$h$ from a function~$f$, given as example.
  $$
  \begin{array}{c|rrrrrr rrrrrrr}
     n &0  &1  &2  &3  &4  &5 &6  &7  &8  &9  &10 &11 &\ldots \\
  f(n) &\rule[-2mm]{0.0mm}{5mm}
        3  &2 &5  &5  &3 &40 
       &\ldots &\ldots &\ldots &\ldots &\ldots &\ldots &\ldots
  \\ \hline
  g(n) &\rule{0.0mm}{4mm}
        \ov{3}  &3  &\ov{2}  &2  &\ov{5}  
       &5         &5 &5  &3  &3  
       &\ov{40} &40&\ldots\\
  h(n) &6  &1  &4  &3 &10 &5  
       &7  &9  &11 &13&80 &15 &\ldots
  \end{array}
  $$
  Whenever the value of~$g(n)$ is new (bar over the number), the value
  of~$h(n)$ is $2\times g(n)$. If~$g(n)$ it occurred before, the successive
  odd integers~1, 3, 5\ldots\ are used as values of~$h(n)$.
  The function~$h$ will always be injective. However, it will be surjective
  only if~$f$ is also surjective.
\end{proof}

The table in Figure~\ref{summary} (page~\pageref{summary}) summarises
our undecidability results about PR functions.  

\begin{figure}
{
\begin{center}
\renewcommand{\tabcolsep}{3mm}
\renewcommand{\arraystretch}{1.5}
$$
\begin{tabular}{|cllc|}                                      \hline
  Num & PR problem    &  Class              & Proof in page  \\ \hline
  \textcolor{dred}{1} &
  \ozN          & $\sigOne$-complete
                      &  \pageref{t-oz}        \\
  \textcolor{dred}{2} &  
  \eozN         & $\deltaTwo\setminus(\sigOne\cup \piOne)$  
                      &  \pageref{t-eoz}       \\
  \textcolor{dred}{3} &  
  \gtkzN         & $\sigOne$-complete    
                      &  \pageref{t-gtkz}     \\
  \textcolor{dred}{4} &  
  \ekzN          & $\deltaTwo\setminus (\sigOne\cup \piOne)$  
                      &  \pageref{t-ekz}      \\
  \textcolor{dred}{5} &  
  \azN           & $\piOne$-complete     
                      &  \pageref{t-az}       \\
  \textcolor{dred}{6} &  
  \infzN        & $\piTwo$-complete   
                      &  \pageref{t-infzCOMP}      \\
  \textcolor{dred}{7} &  
  \kCodON             & $\piOne$-complete                 & 
             \pageref{kCodT}                                   \\
  \textcolor{dred}{8} &  
  \kCodN, $k\geq 2$ & $\deltaTwo\setminus (\sigOne\cup\piOne)$
                      &   \pageref{kCodTwo}   \\
  \textcolor{dred}{9} &  
  \fCodN              & $\sigTwo$-complete  
                      &   \pageref{fCodT}     \\
  \textcolor{dred}{10}&  
  \equiN              & $\piOne$-complete  
                      &   \pageref{t-equi}
                       \\ 
  \textcolor{dred}{11}&  
  $\prinj$            & $\piOne$-complete
                      &   \pageref{Tprinj}                      \\
  \textcolor{dred}{12}&  
  $\prsurj$           & $\piTwo$-complete
                      &  \pageref{tt-surj}                      \\
  \textcolor{dred}{13}&  
  $\prbij$            & $\piTwo$-complete
                      &  \pageref{tt-bij}                 \\ \hline
\end{tabular}
$$
\end{center}
}
\caption{Some decision problems studied in this paper.
  The instance consists of one or two PR functions (two
  for the \equiN\ problem). 
  The problems are numbered (first column) for
  reference in Figure~\ref{classes}, page~\pageref{classes}.
  Note that the question associated
  with the problems in class~8 {\em is not}
  ``is the size of the codomain at least~2?'' but 
  ``is the size of the codomain exactly equal to~$k$?'' (for some fixed $k\geq 2$).}
\label{summary}
\end{figure}

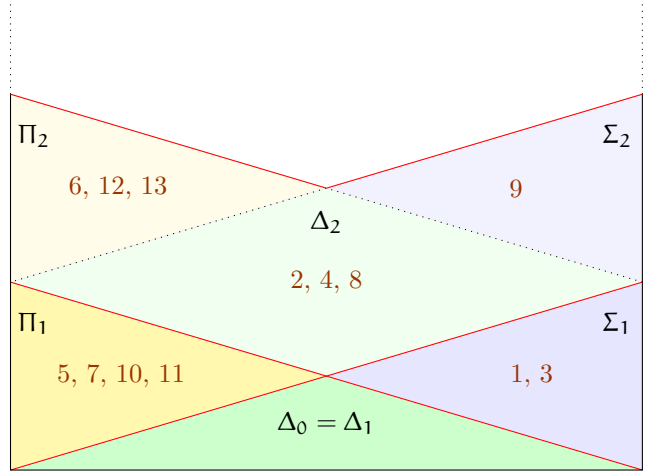
\begin{figure}
\begin{center}
  \begin{tikzpicture}[xshift=2.5cm, yshift=2.0cm, xscale=4.2, 
    yscale=2.5,auto]
  \draw[gray!10, ultra thin, fill=green!20]  
         (0,0) -- (2,0) -- (1,0.5) -- (0,0);
  \draw[gray!10, ultra thin, fill=yellow!40] 
         (0,0) -- (1,0.5) -- (0,1) -- (0,0);
  \draw[gray!10, ultra thin, fill=blue!10]   
         (2,0) -- (2,1) -- (1,0.5) -- (2,0);
  \draw[gray!10, ultra thin, fill=yellow!10]   
         (0,1) -- (1,1.5) -- (0,2) -- (0,1);
  \draw[gray!10, ultra thin, fill=blue!05]   
         (2,1) -- (2,2) -- (1,1.5) -- (2,1);
  \draw[gray!10, ultra thin, fill=green!06]   
         (0,1) -- (1,0.5) -- (2,1) -- (1,1.5) -- (0,1);
  \draw[black, dotted, thin]   
         (2,1) -- (1,1.5) -- (0,1);
  \draw[red,   thin]  (0,2) -- (1,1.5);
  \draw[red,   thin]  (2,2) -- (1,1.5);
  \draw[red,   thin]  (0,1) -- (2,0);
  \draw[red,   thin]  (2,1) -- (0,0);
  \draw[black, thin]  (0,2) -- (0,0) -- (2,0) -- (2,2);
  \draw[black, thin, dotted]    (0,2) -- (0,2.5);
  \draw[black, thin, dotted]    (2,2) -- (2,2.5);
    \node at (1.920,1.78)   {$\sigTwo$};
    \node at (0.075,1.78)   {$\piTwo$};
    \node at (1.000,1.32)   {$\deltaTwo$};
    \node at (1.920,0.78)   {$\sigOne$};
    \node at (0.075,0.78)   {$\piOne$};
    \node at (1.000,0.25)   {$\deltaZero=\deltaOne$};
    \node at (1.65,0.50) {\textcolor{dred}{{1}, {3}}};
    \node at (0.35,0.50) {\textcolor{dred}{{5}, {7}, {10}, {11}}};
    \node at (1.60,1.50) {\textcolor{dred}{{9}}};
    \node at (1.00,1.00) {\textcolor{dred}{2, 4, 8}};
    \node at (0.34,1.50) {\textcolor{dred}{6, 12, 13}};
  \end{tikzpicture}
  \end{center} 
  \caption{Location in the arithmetic hierarchy of the problems
    mentioned in Figure~\ref{summary}, page~\pageref{summary}. The
    problem numbers (1 to~12 in \textcolor{dred}{this} colour) refer to
    the numbers in the first column of Figure~\ref{summary}. 
    All the decision problems in~$\sigOne$, $\piOne$,
    $\sigTwo$, and $\piTwo$ {\em are complete in the respective class.}}
\label{classes}
\end{figure}

\subsection{Primitive recursive versus partial recursive problems}
\label{compare}
It is interesting to compare the degree of undecidability of the same
problem about a function~$f$ in two cases: (i)~$f$ if primitive
recursive, (ii)~$f$ is partial recursive.

\begin{theorem}
  \label{alsoPART}
  Let~$A$ be a PR decision problem corresponding to the question
  ``does a given PR function have the property~$P$?'' and let~$X$ be
  any decision problem.  If $X\reduces A$, then $X\reduces A'$,
  where~$A'$ corresponds to the question ``does a given partial
  recursive function have the property~$P$?''.
\end{theorem}

\begin{proof}
  Let~$f$ be the function corresponding to the reduction $X\reduces
  A$.  The image by~$f$ of a (positive or negative) instance~$x$
  of~$X$ is a PR function, thus it is also a partial function.
  So, we can use the same function~$f$ for the reduction
  $X\reduces A'$.
\end{proof}


\begin{figure}
{
\begin{center}
\renewcommand{\arraystretch}{1.5}
$$
\begin{tabular}{|llcl|}                                       \hline
\multirow{2}{*}{Problem}
           &  Primitive recursive   &   & Partial recursive      \\
           &  (problem $\PRpow{{\rm P}}$) &   & (problem P)          \\ \hline
$f(0)=0$   & $\Delta$ & $<$ & $\sigOne$-complete                  \\
\totalN    & $\Delta$ & $<$ & $\piTwo$-complete                  \\
\recurN    & $\Delta$ & $<$ & $\sigThree$-complete                \\
\GozN       & $\sigOne$-complete &$=$&$\sigOne$-complete         \\
\GnoZerosN  & $\piOne$-complete  & $=$ &
                $\piOne$-complete                                 \\
\GeozN      & $\deltaTwo\setminus(\sigOne\cup \piOne)$         &
       $<$    & $\piTwo$-complete                                 \\
\GgtkzN     & $\sigOne$-complete & $=$ &                        
                $\sigOne$-complete                                \\
\GekzN      & $\deltaTwo\setminus (\sigOne\cup \piOne)$        &
       $<$    & $\piTwo$-complete                                 \\
\GazN       & $\piOne$-complete  & $<$ &                          
                $\piTwo$-complete                                  \\
\GinfzN    & $\piTwo$-complete   & $=$ &                        
               $\piTwo$-complete                                  \\
\fdomN     & $\Delta$            & $<$ &                    
                    $\sigTwo$-complete                            \\
\GkCodON   & $\piOne$-complete & $<$ &                    
                    $\deltaTwo\setminus(\sigOne\cup\piOne)$       \\
\GkCodN, $k\geq 2$ & $\deltaTwo\setminus (\sigOne\cup\piOne)$  & 
       $=$      & $   \deltaTwo\setminus (\sigOne\cup\piOne)$      \\
\GequiN     & $\piOne$-complete & $<$ &$\piTwo$-complete
\\ 
$\Gprinj$   & $\piOne$-complete & $=$                    & 
                $\piOne$-complete                                 \\
$\Gprsurj$  & $\piTwo$-complete & $=$                    &  
                $\piTwo$-complete                                 \\
$\Gprbij$   & $\piTwo$-complete & $=$                    & 
                $\piTwo$-complete                          \\ \hline
\end{tabular}
$$
\end{center}
}
\caption{The degree of undecidability of some decision problems about
  one or two functions.
  For each problem~$P$ two cases are compared: 
  the function is primitive recursive ($\PRpow{{\rm P}}$) and 
  the function is partial recursive (P). The
  symbol~``$=$'' means that the partial recursive problem and the 
  corresponding PR problem
  have the same degree of undecidability (in terms of the arithmetic
  hierarchy), while the 
  symbol~``$<$'' means that the PR problem problem is ``less undecidable'' 
  than the corresponding partial recursive problem.
  The \recurN\ is defined in Example~\ref{Precur}, page~\pageref{Precur}.
}
\label{c-summary}
\end{figure}

\subsubsection*{Partial recursive functions: some undecidability results}
\label{parRec}
We now proof the AH classes mentioned in the last column of
Figure~\ref{c-summary} (page~\pageref{c-summary}).
The instance of the decision problems studied in this section is
thus a {\em ParRec function}.

When we mention the value~$f(x)$ of a ParRec 
function~$f$ it is implicitly assumed that this value exists, 
that is, that $\exists t:F(x,t)$; {\em recall the meaning of the predicate}
$F(x,t)$, see Definition~\ref{interval}, page~\pageref{interval}.
\begin{enumerate}
\item  
  $f(0)=0$: $\sigOne$-complete, easy to reduce the SHP to this problem.
\item  
  \totalN: $\piTwo$-complete, see Theorem~\ref{THREEcomp} 
  (page~\pageref{THREEcomp}).
\item 
  \GozN: $\sigOne$-complete.
\item  \GnoZerosN, $f$ has no zeros: $\piOne$-complete.
\item  \label{eozIT}
  \GeozN: 
  $\piTwo$-complete. \\
  Proof.
  The corresponding statement is below\footnote{Lines {\fnsz\sf \fnsz{(4)}}
  and {\fnsz\sf \fnsz{(5)}}: when we speak of the value
    of~$f(x_1)$ we also say that~$f(x_1)$ is defined; and similarly 
    for~$f(x_2)$.} (lines~{\fnsz~(1)} to~{\fnsz~(5)}). 
  The $\forall\exists$ sequence originates in lines {\fnsz\sf \fnsz{(4)}} 
  and {\fnsz\sf \fnsz{(5)}}.\\
\spac
    {\fnsz (1)}\spac  $\exists x: \text{$[f(x)$ is defined and $f(x)=0]$}$ and\\
\spac
    {\fnsz (2)}\spac  $\forall x_1,\,x_2:$ either\\
\spac
    {\fnsz (3)}\spacc $[x_1=x_2]$ or $[f(x_1)$ is undefined$]$ 
           or $[f(x_2)$ is undefined$]$ or\\
\spac
    {\fnsz (4)}\spacc $\text{$[f(x_1)$ defined and $f(x_1)\neq 0]$}$ or\\
\spac
    {\fnsz (5)}\spacc $\text{$[f(x_2)$ defined and $f(x_2)\neq 0]$}$,

  To prove completeness we reduce the $\piTwo$-complete \GequiN\ 
  (Theorem~\ref{THREEcomp}, page~\pageref{THREEcomp}) to\\
  \GeozN. 
  Consider an instance $\pair{f,g}$ of \GequiN\ and define
\rule[-2mm]{0.0mm}{5mm}

  \noindent
  \begin{minipage}{0.85\linewidth}
  \noindent
  {\tt Function} $h(n)$: \\
  \spac {\tt if $n=0$: return} 0;\\
  \spac  {\tt if $n\geq 1$}:\\
  \spac\spac {\tt if} for some $x\leq n$, $f(x)$ and $g(x)$ are already 
              defined at step~$\leq n$
              with $f(x)\neq g(x)$, {\tt then} \\
  \spac\spac\spac  {\tt return 0}; \spacc {\tt // $f\neq g$}\\
  \spac\spac  {\tt else}\\
  \spac\spac\spac {\tt return 1} \spacc {\tt // $f=g$ until now.}\\
  \end{minipage}

\rule[1mm]{0.0mm}{5mm}
Clearly the function~$h(n)$ has exactly one zero iff~$f$ and~$g$
are the same (partial) function.

\item  
  $\Ggtkz$: $\sigOne$-complete. \\
  It is easy to show that it is in~$\sigOne$
  (use the ``dovetailing'' technique); and that \GozN\ reduces easily to \GgtkzN,
  see the proof of Theorem~\ref{t-gtkz} (page~\pageref{t-gtkz}).
\item \GekzN\ ($k$ fixed): 
  $\piTwo$-complete. \\
  Proof similar to \GeozN, see item~\ref{eozIT}.
\item  
  \GazN: $\piTwo$-complete. \\
  Proof.
  The corresponding question can be expressed as
  $$
  \forall x \exists t: F(x,t) \wedge (f(x)=0)
  $$
  so that \GazN\ belongs to~$\piTwo$.\\
  Consider the $\piTwo$-complete problem \total, Theorem~\ref{THREEcomp}
  (page~\pageref{THREEcomp}). Reduce \totalN\ to \GazN\ as follows.
  Given an instance~$f$ of \totalN, define
  $$
  g(n)=
  \left\{
  \begin{array}{ll}
    0 & \text{\spac if the computation of $f(n)$ converges} \\
    {\rm undefined} 
      & \text{\spac otherwise}
  \end{array}
  \right.
  $$
  The function~$g$ is the zero function iff~$f$ is total.
\item 
  \GinfzN: $\piTwo$-complete.\\
  Proof. The question associated with the problem can be expressed as
  \begin{equation*}
    \forall m\,\exists x, t: 
        (x\geq m) \wedge F(x,t) \wedge (f(x)=0)
  \end{equation*}
  Thus \GinfzN\ belongs to~$\piTwo$. \\
  It is not difficult to reduce $\neg\fdom$ (recall the \fdomN\ is 
  $\sigTwo$-complete, Theorem~\ref{THREEcomp}, page~\pageref{THREEcomp}) 
  to \GinfzN, from which the result follows. 
\item  
  \label{proof-COD-one-ParR}
  \GkCodON: $\deltaTwo \setminus (\sigOne\cup\piOne)$.\\
  Proof. \GkCodON\ can be expressed as {\em the negation} of
  $$
     \overbrace{\neg[\exists x, t : F(x,t)]}^{(1)}
     \;\vee\;
     \overbrace{\exists x_1, x_2, t_1, t_2, :
       F(x_1,t_1) \wedge F(x_2,t_2) \wedge f(x_1)\neq f(x_2)}^{(2)}
  $$
  Part~(1) means ``$f$ it is the totally undefined function'' while
  part~(2) means ``there are at least two values of~$x$ for which~$f(x)$ is
  defined and has different values. We can rewrite this logical
  statement, which is the negation of \GkCodON, as
  $$
     \exists x_1, x_2, t_1, t_2 \; \forall x, t : 
        [\neg F(x,t)]
        \;\vee\;
        [F(x_1,t_1) \wedge F(x_2,t_2) \wedge f(x_1)\neq f(x_2)]
  $$
  In this case the order of the quantifiers can be changed, see the 
  proof of Theorem~\ref{t-eoz} (page~\pageref{t-eoz}).
  Thus \GkCodON\ (whose statement is the negation of the above) 
  belongs to $\deltaTwo=\sigTwo\cap\piTwo$.\\ 
  We now show that the problem \GkCodON\ is neither in~$\sigOne$ 
  nor in~$\piOne$.
  First reduce $\neg${\rm SHP} to (\GkCodON). Let~$n$ be the instance of
  $\neg${\rm SHP}. Define
  $$
  f(t) = 
  \left\{
    \begin{array}{ll}
      0   & \text{\spac if $t=0$}\\
      0   & \text{\spac if $\vph_n(n)$ does not converge in~$\leq t$ steps
            ($t\geq 1$)}  \\
      1   & \text{\spac otherwise ($t\geq 1$)}
    \end{array}
  \right.
  $$
  Clearly~$\vph_n(n)$ diverges iff the codomain of~$f$ has size~1.\\
  Now reduce {\rm SHP} to (\GkCodON). Let~$n$ be the instance of
  {\rm SHP}. Define
  $$
  f(t) = 
  \left\{
    \begin{array}{ll}
      \text{undefined} & 
      \text{\spac if $\vph_n(n)$ does not converge in~$\leq t$ steps}  \\
      0   & \text{\spac otherwise}
    \end{array}
  \right.
  $$
  The computation~$\vph_n(n)$ halts iff the codomain of~$f$ has size~1.
  Thus \GkCodON\ belongs to~$\deltaTwo\setminus(\sigOne\cup\piOne)$.
\todox{Review the proof.}
\item \GkCodN, $k\geq 2$: $\deltaTwo\setminus(\sigOne\cup\piOne)$.\\ 
  Proof. We exemplify for the case~$k=2$, the
  cases~$k\geq 3$ are similar. 
  The the statement \GkCodTN\ can be expressed as
  $$
  \begin{array}{ll}
    \exists x,x',t,t': F(x,t) \wedge F(x',t') \wedge
        (f(x)\neq f(x'))   \;\wedge\; 
               & \text{\spac$\abs{\cod{f}}\geq2$}\\
    \forall x_1,x_2,x_3,t_1,t_2,t_3 : Q(x_1,x_2,x_3,t_1,t_2,t_2)
               & \text{\spac$\abs{\cod{f}}<3$}\\
 \end{array}
  $$
where $Q(x_1,x_2,x_3,t_1,t_2,t_3)$ 
  denotes\footnote{When $F(x_1,t_1)$ is false (the computation of~$f(x_1)$
    has not yet halted) the value of~$f(x_1)$  is irrelevant because the 
    disjunct $\neg F(x_1,t_1)$ is true. 
    In more detail: if we write~$f_{t_1}(x_1)$ instead of~$f(x_1)$,
    the value of~$f_{t_1}(x_1)$ is arbitrary when $F(x_1,t_1)$ is false
    and $f_{t_1}(x_1)=f(x_1)$ otherwise.
    And similarly for~$f(x_2)$ and~$f(x_3)$.}
$$
  \neg F(x_1,t_1) \vee \neg F(x_2,t_2) \vee \neg F(x_3,t_3) \; \vee  
  (f(x_1)=f(x_2)) \vee (f(x_2)=f(x_3)) \vee (f(x_3)=f(x_1))\Mdot
$$
The statement \GkCodTN\ can be rewritten in 2 forms:
$$
\begin{array}{l}
  \exists x,x',t,t' \:\forall x_1,x_2,x_3,t_1,t_2,t_3 :
    [F(x,t) \wedge F(x',t') \wedge (f(x)\neq f(x'))] \wedge
    Q(x_1,x_2,x_3,t_1,t_2,t_3)
\\
  \forall x_1,x_2,x_3,t_1,t_2,t_3 \: \exists x,x',t,t':
    [F(x,t) \wedge F(x',t') \wedge (f(x)\neq f(x'))] \wedge
    Q(x_1,x_2,x_3,t_1,t_2,t_3)\Mdot
\end{array}
$$
(see the proof of Theorem~\ref{t-eoz}, page~\pageref{t-eoz}.)
Thus
\GkCodTN\ belongs to~$\sigTwo\cap\piTwo=\deltaTwo$.

To prove that \GkCodTN\ is not in~$\piOne$,
reduce {\rm SHP} (instance~$n$) to \GkCodTN.
Let~$n$ be the instance of {\rm SHP}.
$$
  f(t) = 
  \left\{
    \begin{array}{ll}
      0   & \spac\text{if $t=0$}\\
      0   & \spac\text{if $\vph_n(n)$ did not converge in~$\leq t-1$ steps
            ($t\geq 1$)}  \\
      1   & \spac\text{if $\vph_n(n)$ converges in~$\leq t-1$ steps
            ($t\geq 1$)}
    \end{array}
  \right.
$$
(the case $f(0)=0$ is considered because the possibility of
convergence for~$t=0$). The function $f$ has codomain with 
size~2 iff $\vph_n(n)\halts$. \\
To prove that \GkCodTN\ is not in~$\sigOne$,
reduce {\rm $\neg$SHP} (instance~$n$) to \GkCodTN.
Let~$n$ be the instance of $\neg${\rm SHP}.
$$
  f(t) = 
  \left\{
    \begin{array}{ll}
      0   & \text{\spac if $t=0$}\\
      1   & \text{\spac if $t=1$}\\
      0   & \text{\spac if $\vph_n(n)$ did not converge in~$\leq t-2$ steps
            ($t\geq 2$)}  \\
      2   & \text{\spac if $\vph_n(n)$ converges in~$\leq t-2$ steps
            ($t\geq 2$)}
    \end{array}
  \right.
$$
The function $f$ has codomain with size~2 iff~$\vph_n(n)$ diverges.
\end{enumerate}

The following problems use Definition~\ref{Ginjsurbij}, page~\pageref{Ginjsurbij}.
\begin{enumerate}
\setcounter{enumi}{11}
\item \GprinjN: $\piOne$-complete.\\
  Proof. For ParRec functions the \GprinjN\
  statement can be expressed as
  \begin{equation}
    \label{EQprinj}
    \neg[ 
    \exists x, y, t, t': 
    F(x,t) \wedge F(y,t') \wedge (f(x)=f(y))
    ]    
  \end{equation}
  Define a reduction from \GozN\ to $\neg$\GprinjN\ as follows. Let~$f$ be
  the instance of \GozN. Define~$g$ as
  $$
    \left\{
    \begin{array}{ll}
      g(0)=0  \\
      g(n)=0 & \text{\spac if $n\geq 1$, $f(n-1)\halts$
                     and $f(n-1)=0$}  \\
      g(n)={\rm undefined} & 
                    \text{\spac otherwise}\Mdot
    \end{array}
    \right.
  $$
  The function~$f$ has at least one~0 iff the function~$g$ {\em is not}
  injective.
  As \GozN\ is $\sigOne$-complete, \GprinjN\ is $\piOne$-complete.
\item  \GprsurjN: $\piTwo$-complete.\\
  Proof. For ParRec functions the \GprsurjN\
  statement can be expressed as
  $$
  \forall y\,\exists x, t: 
  F(x,t) \wedge (f(x)=y)
  $$
  Thus \GprsurjN\ is in~$\piTwo$. \\
  The reduction  from the $\piTwo$-complete problem $\neg\fdom$
  to \prsurjN\ used in the proof of Theorem~\ref{tt-surj} 
  (page~\pageref{tt-surj}) is also a reduction from \\
  $\neg\fdom$ to 
  \GprsurjN\ (because a PR function is also a ParRec function).
\item
\label{GbijP}
  \GprbijN\ (for ParRec functions): $\piTwo$-complete.\\
  Proof. The \GprbijN\ statement can be expressed as
  $$
  \begin{array}{l}
   \forall y, x, x_1, x_2\:\exists x',t:
      \overbrace{F(x,t)}^{\text{total}}                  \;\wedge \;   
      \overbrace{f(x')=y}^{\text{surjective}}                 \;\wedge \;  
      \overbrace{(x_1\neq x_2)\;\Rarr\;f(x_1)\neq f(x_2)}^{\text{injective}}\Mdot
  \end{array}
  $$
  Note that the conjunct ``$\forall x \exists t: F(x,t)$'' (total function)
  ensures that we can talk about~$f(x')$, $f(x_1)$ and~$f(x_2)$ 
  without including any other convergence conditions. 
  Thus \GprbijN\ is in~$\piTwo$. \\
  Recall that the problem \prsurjN\ is $\piTwo$-complete\\
  (reduction $\neg\fdom\reduces\prsurj$, see the proof of
  Theorem~\ref{tt-surj}, page~\pageref{tt-surj})
  and note that the reduction $\prsurj\reduces\prbij$ (proof of 
  Theorem~\eqnref{tt-bij}, page~\pageref{tt-bij}) is also a reduction
  $\prsurj\reduces\Gprbij$. Thus, \GprbijN\ is in $\piTwo$~complete.
\end{enumerate}

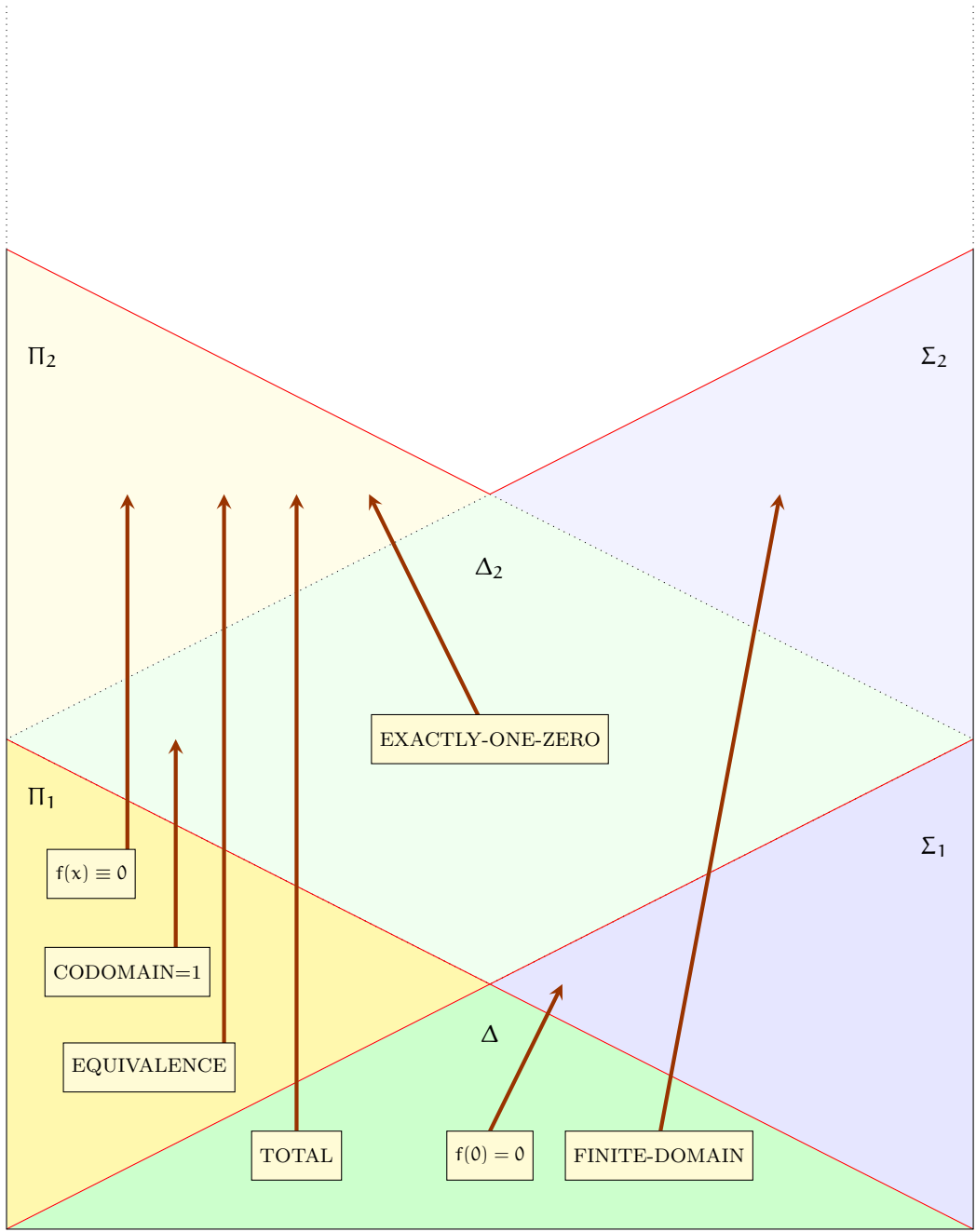
\begin{figure}
\begin{center}
  \begin{tikzpicture}[xshift=2.5cm, yshift=2.0cm, 
    xscale=6.9, yscale=3.5,auto]
  \tikzset{others/.style={rectangle, draw=black, 
           fill=yellow!20, minimum size=0.7cm}}
  \draw[gray!10, ultra thin, fill=green!20]  
         (0,0) -- (2,0) -- (1,1.0) -- (0,0);
  \draw[gray!10, ultra thin, fill=yellow!40] 
         (0,0) -- (1,1.0) -- (0,2) -- (0,0);
  \draw[gray!10, ultra thin, fill=blue!10]   
         (2,0) -- (2,2) -- (1,1.0) -- (2,0);
  \draw[gray!10, ultra thin, fill=yellow!10]   
         (0,2) -- (1,3.0) -- (0,4) -- (0,2);
  \draw[gray!10, ultra thin, fill=blue!05]   
         (2,2) -- (2,4) -- (1,3.0) -- (2,2);
  \draw[black, thin, dotted, fill=green!06]   
         (0,2) -- (1,1.0) -- (2,2) -- (1,3.0) -- (0,2);
  \draw[red,   thin]  (0,4) -- (1,3.0);
  \draw[red,   thin]  (2,4) -- (1,3.0);
  \draw[red,   thin]  (0,2) -- (2,0);
  \draw[red,   thin]  (2,2) -- (0,0);
  \draw[black, thin]  (0,4) -- (0,0) -- (2,0) -- (2,4);
  \draw[black, thin, dotted]    (0,4) -- (0,5.0);
  \draw[black, thin, dotted]    (2,4) -- (2,5.0);
    \node at (1.920,3.56)   {$\sigTwo$};
    \node at (0.075,3.56)   {$\piTwo$};
    \node at (1.920,1.56)   {$\sigOne$};
    \node at (0.075,1.76)   {$\piOne$};
    \node at (1.000,0.80)   {$\Delta$};
    \node at (1.000,2.70)   {$\deltaTwo$};
    \draw[ultra thick,->,>=stealth,dred]  (0.60,0.40) -- (0.60,3.0);  
    \draw[ultra thick,->,>=stealth,dred]  (1.00,0.40) -- (1.15,1.00); 
    \draw[ultra thick,->,>=stealth,dred]  (1.35,0.37) -- (1.6,3.00);  
    \draw[ultra thick,->,>=stealth,dred]  (0.35,1.00) -- (0.35,2.00); 
    \draw[ultra thick,->,>=stealth,dred]  (0.25,1.50) -- (0.25,3.00); 
    \draw[ultra thick,->,>=stealth,dred]  (0.45,0.60) -- (0.45,3.0);  
    \draw[ultra thick,->,>=stealth,dred]  (1.00,2.00) -- (0.75,3.0);  
    \node[others] at (0.175,1.45)   {\footnotesize $f(x)\equiv 0$};
    \node[others] at (1.350,0.30)   {\footnotesize {\rm FINITE-DOMAIN}};
    \node[others] at (1.000,0.30)   {\footnotesize $f(0)=0$};
    \node[others] at (0.250,1.05)   {\footnotesize {\rm CODOMAIN=1}};
    \node[others] at (0.295,0.66)   {\footnotesize {\rm EQUIVALENCE}};
    \node[others] at (0.600,0.30)   {\footnotesize \rm TOTAL};
    \node[others] at (1.000,2.00)   {\footnotesize\GeozN};
  \end{tikzpicture}
  \end{center} 
  \caption{The AH undecidability classes of some decision problems, considered
    in two cases: when the instance is
    a {\em primitive recursive function} (origin of the arrow) and when it
    is a {\em partial recursive function} (tip of the arrow). 
    {\rm CODOMAIN=1}
    denotes \kCodON, 
    and $f(x)\equiv 0$ denotes \azN.
    Whenever the head or the tip of an arrow is in~$\sigOne$ or $\piOne$
    or $\piTwo$, then the corresponding problem is {\em complete} in the 
    corresponding class. See also Figure~\ref{c-summary},
    page~\pageref{c-summary}.
    The decision problem \recurN\ (not represented here)
    has the PR version in~$\Delta$, while the ParRec version is complete
    in~$\sigThree$; see Example~\ref{Precur} page~\pageref{Precur}.
}
\label{changes}
\end{figure}
\section{Frontiers of decidability: the \ozN\ problem}
\label{composition}
This section and the next one are about {\em classes} of decision problems.
In this section we consider the basic problem 
``does the PR function~$f$ have at least one zero?''
and restrict the set of inputs and the set of outputs of~$f$, 
as explained below.
We will find necessary and sufficient conditions for decidability of
the problem ``does the {\em restricted} function have at least one
zero?''.

We restrict the possible inputs of a PR function~$f$ by
pre-applying to~$f$ some PR function\footnote{We use a sans serif
  font for {\em fixed} PR functions like~\tfxg\ and~\tfxh.
  Our results are also valid for fixed (total) recursive functions.}~$\fxg$.
The outputs of~$f$ are similarly restricted by post-applying to~$f$
some fixed PR function~$\fxh$. Both restrictions~-- the input
and the output~-- can be applied simultaneously, see Figure~\ref{proPos}
(page~\pageref{proPos}). Thus, we will study the decidability of a
restricted form of the \ozN\ problem, namely the existence of zeros of
the function~$\fxh(f(\fxg(x)))$.
A more general form of placing a PR function~$f$ in a fixed system
consisting of PR functions is studied in Section~\ref{plug},
page~\pageref{plug}.

In summary, we study the following classes of PR decision problems,
where~$\ov{x}$ denotes the sequence of arguments 
$x_1$, $x_2$\ldots$x_n$ ($n\geq 1$),
$\fxg$ is a fixed PR multifunction (Definition~\ref{multif},
page~\pageref{multif}), and~$\fxh$ is a fixed PR function:
\begin{enumerate}[(a)]
\item \label{pFG} ``$\exists\,\ov{x}: {f}(\fxg(\ov{x}))=0$?'';
\item \label{pHF} ``$\exists\,\ov{x}: \fxh({f}(\ov{x}))=0$?'';
\item \label{pHFG} ``$\exists\,\ov{x}: \fxh(f(\fxg(\ov{x})))=0$?''; 
\item \label{acycHFG} ``$\exists\,\ov{x}: \fxh(\ov{x},f(\fxg(\ov{x})))=0$?''.
\end{enumerate}
The general class~(\ref{acycHFG}) will be studied in Section~\ref{plug}, see 
Theorem~\ref{nft} (page~\pageref{nft}).

Consider for example the function
$\fxg(x)=\lamb{x}{2x^2+1}$ and the question
``Given the PR function~$f$, $\exists x:f(\fxg(x))=f(2x^2+1)=0$?''. 
This problem belongs to the Class~1 above, see Section~\ref{Pfh} 
(page~\pageref{Pfh}).

We will see that the conditions for the decidability of~(c) are simply
the conjunction of the decidability conditions for~(a) and for~(b). In
this sense, the input and the output restrictions are independent.

\subsection*{Why study these families of problems?}
\label{whyComp}
There are several reasons for studying the composition of the instance~$f$
with other PR functions. 
For instance, many atomic predicates can be written as an
equality of the form $F(\ov{x})=0$, where~$F(\ov{x})$ is obtained by
composing~$f$ with fixed PR functions. As an example, the question
``$f(x)=c$?'', where~$c$ is a constant, can be expressed as
``$(f(x)\dminus c)+(c\dminus f(x))=0$?''  which is
equivalent to ``$\fxh({f}(x))=c$?'' where 
$\fxh(z)=\lamb{z}{(z\dminus c)+(c\dminus z)}$.
However, our main interest in this section is to look to a well known
undecidable problem, namely \ozN, and restrict the inputs and outputs
of the function being studied, in order to know ``when the
problem become decidable''. This {\em decidability frontier} is given
by Theorems~\ref{ozFH-t}, \ref{ozHF-t},
and~\ref{ozHFG-t}, respectively in pages~\pageref{ozFH-t},
\pageref{ozHF-t}, and~\pageref{ozHFG-t}. 

\subsection{The meaning of restricted \ozN\ problems}
\label{meaningComp}
The class of decision problems we study in this section is

\bigskip
\begin{minipage}{140mm}\
  \noindent\spacc Problem specification: fixed PR functions $\fxg$ and $\fxh$.\\
  \noindent\spacc \inst PR function $f$.\\
  \noindent\spacc \quest $\exists \ov{x}:\fxh(f(\fxg(\ov{x})))=0$?
\end{minipage}
\bigskip

Recall that the unrestricted case of the problem, namely \ozN, is 
undecidable, Theorem~\ref{t-oz} (page~\pageref{t-oz}).
If we restrict the input of~$f$ by pre-applying a PR function~\tfxg,
the problem remains undecidable if and only if the
set of possible inputs of~$f$ (the codomain of~$\fxg$) is infinite, see
Theorem~\ref{ozFH-t}, page~\pageref{ozFH-t}.  
The output of~$f$ is restricted by post-applying some other PR
function~$\fxh$.  As Theorem~\ref{ozHF-t} (page~\pageref{ozHF-t}) shows,
the problem remains undecidable if and only if the codomain of~$\fxh$
includes~0 and some nonzero integer.
The necessary and sufficient conditions
for the undecidability of this problem are illustrated the diagram of 
Figure~\ref{proPos}. We summarise the conclusions of this section:
\begin{enumerate}
\item [1.] Each of the following conditions trivially implies
  decidability: (i)~the set of possible outputs of~$\fxg$ is finite, 
  (ii)~the output of~$\fxh$ is always (that is, for every input~$\fxg(\ov{x})$ 
  of~$f$) nonzero, (iii)~the output of~$\fxh$ is always zero.
\item [2.] {\em If neither of these conditions is satisfied, the problem is
  undecidable}.
\end{enumerate}

We can use the \ozFG, \ozHF, and \ozHFG\ classes of problems
(Theorems~\ref{ozFH-t}, \ref{ozHF-t}, and \ref{ozHFG-t} respectively)
to prove the undecidability of many decision problems.

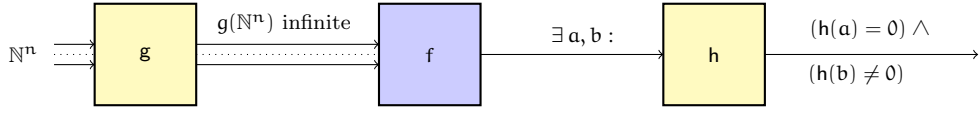
\begin{figure}
{\footnotesize
\begin{center}
  \begin{tikzpicture}[xshift=2.5cm, yshift=2.0cm, scale=2.7, auto,
    main node/.style={circle,fill=blue!20,draw}]
    \draw [black, thick, fill=yellow!30] (0.35,0) rectangle (0.85,0.5);
    \draw [black, thick, fill=blue!20]  (1.75,0) rectangle (2.25,0.5);
    \draw [black, thick, fill=yellow!30] (3.15,0) rectangle (3.65,0.5);
    \node at (0.0,0.25) {$\ene^n$};
    \node at (0.6,0.25) {$\fxg$};
    \node at (2.0,0.25) {$f$};
    \node at (3.4,0.25) {$\fxh$};
    \node at (1.28,0.4)   {$g(\ene^n)$ infinite};
    \node at (2.75,0.35)  {$\exists\,a,b:$};
    \node at (4.18,0.37)  {$(\fxh(a)=0)\;\wedge\;$};
    \node at (4.10,0.15)  {$(\fxh(b)\neq 0)$};
     \draw [->]     (3.65,0.25) -- (4.70,0.25);     
     \draw [->]     (2.25,0.25) -- (3.15,0.25);  
     \draw [->]     (0.85,0.20) -- (1.75,0.20);  
     \draw [->]     (0.85,0.30) -- (1.75,0.30);  
     \draw [dotted] (0.85,0.25) -- (1.75,0.25);  
     \draw [->]     (0.15,0.20) -- (0.35,0.20);  
     \draw [->]     (0.15,0.30) -- (0.35,0.30);  
     \draw [dotted] (0.15,0.25) -- (0.35,0.25);  
  \end{tikzpicture}
  \end{center}
}
\caption{Restricting the \ozN\ problem. The necessary and sufficient
  conditions for undecidability of the global problem 
  ``$\exists\ov{x}:\fxh(f(\fxg(\ov{x})))=0$?''
  are (see the arrow labels):
  ``$[\fxg(\ene^n)$ is infinite $]\;\wedge\;
    [\exists\,a,\,b:
     \fxh(a)=0 \wedge \fxh(b)\neq 0]$.''
  It is assumed that the
  output of~$\fxg$ may be a tuple of integers (Definition~\ref{multif},
  page~\pageref{multif}).}
\label{proPos}
\end{figure}

\bigskip

\noindent{\bf Note.} In the statements of Theorems~\ref{ozFH-t}
(page~\pageref{ozFH-t}), \ref{ozHF-t} (page~\pageref{ozHF-t}),
and \ref{ozHFG-t} (page~\pageref{ozHFG-t})
we can replace ``undecidable'' by ``$\sigOne$-complete''.
This follows directly from the corresponding proofs.
\close

\bigskip

\noindent{\bf Note.} It would have been possible to prove first
Theorem~\ref{ozHFG-t} (page~\pageref{ozHFG-t}), and then present
Theorems~\ref{ozFH-t} (page~\pageref{ozFH-t}) and~\ref{ozHF-t}
(page~\pageref{ozHFG-t}) as corollaries. However,
we prefer to prove the simpler results first.
\close

\subsection{The \ozFG\ class of problems}
\label{Pfh} 
Each problem in the class is specified by the PR function~$\fxg$.
\begin{problem} (\ozFG)
Let~$\fxg$ be a fixed PR function.\\
\inst PR function~$f$. \\
\quest $\exists\,\ov{x}: {f}(\fxg(\ov{x}))=0$?\closex
\end{problem}

As we will see in Theorem~\ref{ozFH-t} (page~\pageref{ozFH-t}), the
decidability of the problem depends on whether the codomain of~$\fxg$ is
infinite.

We begin by proving directly the undecidability of a
particular problem of this class.

\begin{example}
  Let $\fxg(x)=2x+1$. Reduce \ozN\ to \ozFG.  Let~$u$ be an instance of the 
  \ozN\ problem. Define the function~$f$ (instance of \ozFG) as
  $
  f(x)=u(\lfloor x/2\rfloor)
  $. 
  We have 
  $f(\fxg(x))=f(2x+1)=u(\lfloor (2x+1)/2\rfloor)=u(x)$, 
  so that~$u(x)$ has at least one zero iff~$f(x)$
  is a positive instance of the ``$f(\fxg(x))$ problem'' for $\fxg(x)=2x+1$.
\end{example}

Instead of studying more particular cases we now consider 
a class of problems,
proving an undecidability result. For this 
purpose we will have to look more closely to the Kleene Normal 
Form Theorem~\cite{Rogers}.
\begin{theorem}
  \label{ozFH-t}
  The problem \ozFG\ is undecidable if and only if the codomain of~$\fxg$
  is infinite.
\end{theorem}

\begin{proof}
  The following is a decision procedure when the codomain of~$\fxg$ is
  finite, say $\{y_1,\ldots,y_n\}$; compute
  $f(y_1)$,\ldots, $f(y_n)$;  answer {\sc yes} iff any of these
  integers is~0, answer {\sc no} otherwise.

  Suppose now that the codomain of~$\fxg$ is infinite. It is a
  consequence of the Kleene Normal Form Theorem (see Theorem~\ref{Knft},
  page~\pageref{Knft}) that any ParRec function~$f(\ov{x})$ can be written as
  $f(\ov{x})=U(\ov{x},\mu_t:(T(\ov{x},t)))$ where
  \begin{itemize}
  \item [--] The PR function~$T$ ``evaluates~$y$
    steps'' of the computation~$f(\ov{x})$.
    At the end, it returns~0
    if the computation has already finished, and~1 otherwise.
    We assume that if $T(\ov{x},t)$ is~0 for some~$t$, it is also~0 for 
    every~$t'\geq t$.
  \item [--] The PR function~$U$ is similar to~$T$, but when it is
    ``called''~-- if it is ``called''~-- the value of the last
    argument, $\mu_t:T(\ov{x},t)$, guarantees that the computation has
    halted; it then returns the result of the
    computation~$f(\ov{x})$.
  \end{itemize}
  As~$\fxg$ has an infinite codomain, the sequence~$\fxg(0)$,
  $\fxg(1)$\ldots\ contains an increasing infinite sub-sequence
  $\fxg(y_0)<\fxg(y_1)<\ldots$ (with $y_0<y_1<\ldots$), so that
  \begin{equation}
  \label{inproof}
  f(\ov{x})
       = U(\ov{x},\,\mu_t:T(\ov{x},t))
       = U(\ov{x},\,\mu_t:T(\ov{x},\fxg(t)))
  \end{equation}
  This holds because there are~$\fxg(y)$'s arbitrarily
  large. The arguments~$\ov{x}$ can be
  ``included'' in the function~$T$; let~$T_{\ov{x}}(\fxg(y))$ be the resulting
  (PR) function.

  The HP (unary halting problem) can now be reduced to the \ozFG\
  problem.  Let $\langle f,\,\ov{x}\rangle$ be the instance of HP. The
  corresponding instance of
  \ozFG\ is the~$T_{\ov{x}}$ function, constructed from~$f$ and~$\ov{x}$
  as described above. The computation~$f(\ov{x})$ halts iff there is 
  some~$y$ such that $T_{\ov{x}}(\fxg(y))=0$.
\end{proof}

\subsection{The \ozHF\ class of problems}
\label{Phf}
Each problem is specified by the PR function~$\fxh$.

\begin{problem} (\ozHF)  
Let~$\fxh$ be a fixed PR function.\\
\inst  PR function~$f$.   \\
\quest\ $\exists\,\ov{x}: \fxh({f}(\ov{x}))=0$?\closex
\end{problem}

\begin{theorem}
  \label{ozHF-t}
  The \ozHF\ problem is undecidable if and only if~$\fxh^{-1}(0)$ is 
  neither~$\emptyset$ nor~$\ene$.
\end{theorem}

\begin{proof}
  If~$\fxh^{-1}(0)$ is the empty set, $\fxh(f(\ov{x}))\neq 0$ for
  every~$f$ and for every~$\ov{x}$, so that, for every instance~$f$,
  the answer to ``$\exists\,\ov{x}: \fxh({f}(\ov{x}))=0$?'' is 
  {\sc no}.
  Similarly, if $\fxh^{-1}(0)=\ene$, the answer is
  always~{\sc yes}.  In both cases the problem is decidable.

  Suppose now that there are two integers~$a$ 
  and~$b$ with~$\fxh(a)=0$
  and~$\fxh(b)\neq 0$.  We reduce \ozN\ to \ozHF. 
  Let~$f$ be an instance of \ozN, and define the function~$f'(\ov{x})$
  $$
    f'(\ov{x})=
    \left\{
      \begin{array}{lll}
        a & \text{\spac if $f(\ov{x})=0$}         & \spac
                 (\Rightarrow\; \fxh(f'(\ov{x}))=0)   \\
        b & \text{\spac if $f(\ov{x})\neq 0$}     & \spac
                 (\Rightarrow\; \fxh(f'(\ov{x}))\neq0)
      \end{array}
    \right.
  $$
  The function~$f'$ is also PR, as it may be obtained from~$f$ by the
  additional instruction
  ``{\tt if} $f(\ov{x})=0$ {\tt then} $f'(x)=a$ {\tt else} $f'(x)=b$''.
  Clearly, $f(\ov{x})$ has at least one zero iff $\fxh(f'(\ov{x}))$ 
  has at least one zero.
\end{proof}

\subsection{The \ozHFG\ class of problems}
\label{Pgfh}
Each problem is specified by the PR functions~$\fxg$ and~$\fxh$.
\begin{problem} (\ozHFG)  
Let~$\fxg$ and~$\fxh$ be fixed PR functions.\\
\inst  PR function~$f$.  \\
\quest\ $\exists\,\ov{x}: \fxh({f}(\fxg(\ov{x})))=0$?\closex
\end{problem}

\begin{theorem}
  \label{ozHFG-t}
  Let~$\codg$ be the codomain of~$\fxg$.
  The \ozHFG\ problem is undecidable if and only if the following two 
  conditions hold:
  (i)~$\codg$ is infinite, 
  (ii)~$\fxh^{-1}(0)$ is neither the empty set nor~$\ene^n$.
\end{theorem}

\begin{proof}
  $\fxg$ and~$\fxh$ are fixed PR functions.
  If the codomain of~$\fxg$ is finite, or $\fxh^{-1}(0)=\emptyset$, or
  $\fxh^{-1}(0)=\ene^n$, the problem is clearly decidable.

  Suppose that~(i) and~(ii) hold. 
  Reduce the halting problem to \ozHFG. Let~$u(x)$ be the instance of the
  halting problem and let~$T(x,t)=T_{x}(t)$ be the Turing machine that corresponds 
  to the computation of~$u(x)$, see the proof of Theorem~\ref{ozFH-t}, 
  page~\pageref{ozFH-t}.

  Assuming~(i), the computation $u(x)$ halts iff there is some~$t$ such that 
  $T_{x}(\fxg(t))=0$.
  Assuming~(ii), let~$\fxh(a)=0$ and~$\fxh(b)\neq 0$. Recall that~$T_{x}(t)$ 
  is either~1 (computation not yet halted at step~$t$) or~0
  (computation already finished). 
  Define~$T'_{x}(t)$ as
  $$
  T'_{x}(t) =
  \left\{
    \begin{array}{ll}
      a  &  \text{\spac if $T_{x}(t)=1$}  \\
      b  &  \text{\spac if $T_{x}(t)=0$}
    \end{array}
  \right.
  $$
  Then, $\fxh(T'_{x}(t))=0$     if $T_{x}(t)=1$ and
        $\fxh(T'_{x}(t))\neq 0$ if $T_{x}(t)=0$.
  Thus~$u(x)$ halts iff $v_{x}(\fxg(t))$ 
  has at least one zero, where $v_{x}(t)=\fxh(T'_{x}(t))$.
  The transformation $u(x)\to T'_{x}(t)$ defines the reduction 
  \HP$\reduces$\ozHFG.
  In summary,
  \begin{center}
    \begin{tabular}{clclcl}
             & $u(x)$ halts                        &
      $\lra$ & $\exists t: T(x,t)=0$         &
      $\lra$ & $\exists t: T_{x}(t)=0$         \\
      $\lra$ & $\exists t: T_{x}(\fxg(t))=0$           &
      $\lra$ & $\exists t: T'_{x}(\fxg(t))=a$     &
      $\lra$ & $\exists t: \fxh(T'_{x}(\fxg(t)))=0$
    \end{tabular}
  \end{center}
\end{proof}





\subsection{Classifying a problem~-- 
testing the properties of~$\fxg$ and~$\fxh$}
\label{testGH}
Suppose that we are given a problem in one of the classes \ozFG, 
\ozHF\ or \ozHFG, and that we want to know if the problem
is decidable or not. Perhaps not surprisingly, this ``classifying problem''
is itself undecidable.

--~Consider first the class of problems \ozFG. 
In order to classify an arbitrary problem of this class, we have to
solve the meta-problem, 
``is the codomain of a given PR function~$\fxg$ infinite?''.
However, as we have seen in Theorem~\ref{fCodT} (page~\pageref{fCodT}) this
problem is undecidable, belonging to the class
$\piTwo$-complete.

--~Consider now the class \ozHF\ of decision problems. 
The corresponding question associated with the post-function~$\fxh$ 
can be expressed as the following problem

\begin{problem}
\zeroMoreN   \\
\inst  PR function~$\fxh$.  \\
\quest\ $\exists\, a,\,b: 
  \fxh(a)=0$ and $\fxh(b)\neq 0$?\closex
\end{problem}

\begin{theorem}
  The problem \zeroMoreN\ is complete in the class~$\sigOne$.
\end{theorem}

\begin{proof}
  Membership in~$\sigOne$ is obvious. Reduce the $\sigOne$-complete
  problem \ozN\ to \zeroMoreN. Let~$f$ be the instance of \ozN, and let
  $$
  \fxg(x) = 
  \left\{
    \begin{array}{ll}
      1   & \text{\spac if $x=0$}                      \\
      0   & \text{\spac if $x\geq 1$ and $f(x-1)=0$}     \\
      1   & \text{\spac if $x\geq 1$ and $f(x-1)\neq 0$}
    \end{array}
  \right.
  $$ 
  Clearly, $f$ has at least one zero iff there are integers~$a$
  and~$b$ such that $\fxg(a)=0$ and $\fxg(b)\neq 0$.
\end{proof}

Thus if the PR function~$\fxh$ for the class \ozHF\ is
given by a recursive definition or by Loop program, it is not decidable
whether the condition expressed in Theorem~\ref{ozHF-t}
(page~\pageref{ozHF-t}) is satisfied.

--~For the class of problems \ozHFG\ we get a similar 
``undecidability of classification''.

In conclusion, given a problem~$P$ belonging to one of the classes
\ozFG, \\
\ozHF, or \ozHFG\ (the corresponding instance is~$\fxg$,
$\fxh$ or $\ang{\fxg,\fxh}$),
the classification of~$P$ as decidable or undecidable is itself an
undecidable problem.

\section{Generalisations of the PR decision problems}
\label{geni}
We now study two generalisations of the $\fxh(f(\fxg(x)))$ problem:
\begin{itemize}
\item [--] The PR function~$f$ (the instance of the decision problem) 
  is included in an arbitrary acyclic PR system~$S$.
  See Section~\ref{plug} (page~\pageref{plug}).
\item [--] More than one occurrence of~$f$ is possible. In particular,
  there is the possibility of composing~$f$ with itself.
  See Section~\ref{itself} (page~\pageref{itself}).
\end{itemize}
As previously, the instance of each problem is always be a PR
function~$f$, while~\tfxg, $\fxg'$, \tfxgi{1}, \tfxgi{2}\ldots, \tfxh\
denote fixed PR functions.

\subsection{A primitive recursive function~$f$ in an acyclic PR structure: 
normal form}
\label{plug}
In Section~\ref{Pgfh} (page~\pageref{Pgfh}) we fully characterised the
existence of zeros of a function with the form~$\fxh(f(\fxg(\ov{x})))$,
where~$f$ is the instance of the problem and~$\fxg$ and~$\fxh$ are
fixed PR functions. We consider now a more general situation in which
there is a unique occurrence of~$f$ ``inside'' an arbitrary acyclic
primitive recursive structure\footnote{Or, more generally, inside an
arbitrary acyclic (total) recursive structure.}, see Section~\ref{plug} (see also
Figure~\ref{inside}, page~\pageref{inside}). For any such structure
there is an equivalent normal form, see Theorem~\ref{nft}
(page~\pageref{nft}). 

We begin by defining acyclic PR expressions.
\begin{definition}
  \label{agd}
  An {\em acyclic PR expression} {\rm (acyclic-PR-exp)} is an acyclic
  directed graph characterised by:
  \vspace{-\topsep}
  \begin{enumerate}[--]
  \item nodes: input variables~$x_i$, inputs of the PR functions,
    outputs of the PR functions, and the output variable~$y$;
  \item edges: $(a,b)$ where~$a$ is either 
    an input variable or the output of a PR function, and~$b$ is 
    either the input of a PR function or~$y$;
  \item the inputs of PR functions have indegree~1,
    the output of PR functions have positive outdegree,
    $y$ has indegree~1 and outdegree~0,
    the input variables have indegree~0.
  \end{enumerate}
  \vspace{-\topsep}
  One PR function is~$f$. The other PR functions are fixed.
  There are no loops in the graph.
\close
\end{definition}
A particular acyclic-PR-exp is illustrated in Figure~\ref{ag}, page~\pageref{ag}.

An acyclic-PR-exp can be expanded into a single formula. In the example of
Figure~\ref{ag}, we get
$$
   m(f(x_1,q(x_1,x_2)), q(x_1,x_2), p(q(x_1,x_2),x_2))\Mdot
$$
Sub-expressions may be repeated; in this example,
``$q(x_1,x_2)$'' occurs~3 times. Another way of representing an acyclic-PR-exp
is to use {\em definitions} in a system of equations. This method avoids
the repetition of sub-expressions. In our example
(Figure~\ref{ag}, page~\pageref{ag}) a possible system of equations is
$$
\left\{
  \begin{array}{lcl}
    z &=& q(x_1,x_2)   \\
   y' &=& f(x_1,z)     \\
    v &=& p(z,x_2)     \\
    y &=& m(y',z,v)\Mdot
  \end{array}
\right.
$$
A system of equations equivalent to an acyclic-PR-exp must satisfy the following
conditions
\vspace{-\topsep}
\begin{enumerate}
\item [--] Every left hand side is the definition of a {\em new}
  variable, that is, a variable that does not occur in previous
  equations nor on the right hand side of this equation.
\item [--] The variable defined by the last expression is the output
  variable.
\item [--] Each right hand side has the form $s(w_1,\ldots,w_n)$
  where~$s$ is a function and
  each~$w_i$ is either an input variable or a previously
  defined variable.
\end{enumerate}
\vspace{-\topsep}
\subsubsection{The normal form of a acyclic PR expressions}
\label{normalf}
Suppose that a given PR function~$f$ is placed 
somewhere inside a fixed acyclic PR structure (or acyclic-PR-exp)~$S$. 
Consider the decision problem
\begin{problem} \ \\
\inst  A PR function~$f$\\
\quest Does the function~$S(f,\ov{x})$ have at least a zero?\closex
\end{problem}

In this general case, and contrarily to the
problem~$\fxh(f(\fxg(\ov{x})))=0$ (Theorem~\ref{ozHFG-t},
page~\pageref{ozHFG-t}) it does not seem possible to get a closed
condition for the decidability of the problem.

As a step towards the analysis of this problem, we will show
that such an arbitrary acyclic structure can be reduced to a normal form
illustrated in Figure~\ref{nf}, page~\pageref{nf}. This
normalisation has nothing to do with primitive recursive
functions and can be applied to an arbitrary acyclic graph of
functions that includes a distinguished function~$f$.

\label{Sproblem}
Consider Figure~\ref{relat}, page~\pageref{relat},
which represents the condition ``$\exists x:S(f,\ov{x})=0$?'', {\em
  where~$S$ has been reduced to the normal form} (Theorem~\ref{nft}
below). We can see the origin of the difficulty of finding solutions
of~$S(f,\ov{x})=0$, even when~$S$ is in the normal form: there must be a
pair $\langle \ov{x},y'\rangle$ in the set~$\fxh^{-1}(0)$ such that,
if we apply~$\fxg$ and~$f$ in succession to~$\ov{x}$, we get that same
value~$y'$. That is,
$$
  \exists \overbrace{\langle \ov{x},y'\rangle \in \fxh^{-1}(0)}^{(1)} :
  \overbrace{f(\fxg(\ov{x}))=y'}^{(2)}\Mdot
$$
We now have two interacting conditions ((1) and~(2) above) for the
existence of a zero, and this is the reason of the difficulty in
finding a general explicit condition for the existence of zeros
of~$S(f,\ov{x})$.

\begin{figure}[t]  
  \begin{center}
  \begin{tikzpicture}[xshift=2.0cm, yshift=2.5cm, scale=1.8, rotate=90,auto,
    main  node/.style={rectangle,fill=blue!20,draw},
    fixed node/.style={rectangle,fill=yellow!20,draw},
    input node/.style={}
   ]
    \node[input node] (x2) at (0,4) {$x_2$};
    \node[input node] (x1) at (2,4) {$x_1$};
    \node[input node, inner sep=0, outer sep=0] (y1) at (1,3) {};
    \node[fixed node] (m)  at (1,1) {$m$\rule[-2mm]{0.0mm}{5mm}};
    \node[input node] (y)  at (1,0) {$y$};
    \node[fixed node] (q)  at (1,3) {$q$\rule[-2mm]{0.0mm}{5mm}};
    \node[fixed node] (p)  at (0,2) {$p$\rule[-2mm]{0.0mm}{5mm}};
    \node[main  node] (f)  at (2,2) {$f$\rule[-2mm]{0.0mm}{5mm}};
    \draw [->, transform canvas={yshift=1.5mm}] (x1) to node {} (f);
    \draw [->, transform canvas={yshift=-1.5mm}] (x2) to node {} (p);
    \draw [->] (x1) to node {} (q);
    \draw [->] (x2) to node {} (q) ;
    \draw [->] (q)  to node [above=-0.75mm] 
         [pos=0.15]{$\vertSymb{$z$}{-20mm}$} (m);
    \draw [->] (28.5pt,82.0pt)  to node {} (p.150);
    \draw [->] (28.5pt,82.0pt)  to node {} (f.220);
    \draw [->] (p) +(down:1.1mm)   to node [swap] {$v$} (m.210);
    \draw [->] (f) +(down:1.1mm) to node {$y'$} (m.150);
    \draw [->] (m) to node {} (y);
  \end{tikzpicture}
  \end{center}
  \caption{An example of an {\em acyclic PR expression}, or ``acyclic-PR-exp''.
  In general, the corresponding closed output expression contains 
  multiples occurrences of the same sub-expression. In this case
  the closed output expression is
  $y=m(f(x_1,q(x_1,x_2)), q(x_1,x_2), p(q(x_1,x_2),x_2))$
  in which the sub-expression $q(x_1,x_2)$ occurs three times.}
\label{ag}
\end{figure}
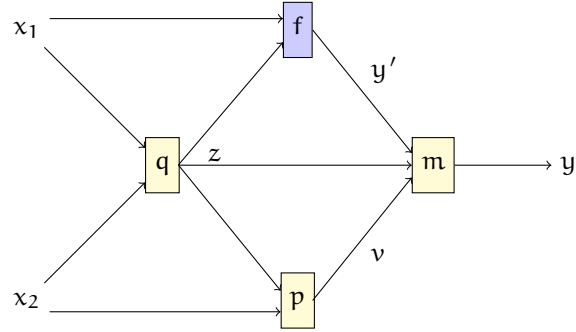

The main result of this section is the following theorem.
\begin{theorem}[Normal Form Theorem]
\label{nft}
An arbitrary acyclic PR expression $S(f,\ov{x})=y$ containing an
occurrence of~$f$ can be reduced to the form
$\fxh(\ov{x},f(\fxg(\ov{x})))$, where~$\fxg$ is a fixed (not depending
on~$f$) PR multifunction (Definition~\ref{multif},
page~\pageref{multif}) and~$\fxh$ is a fixed PR function.
The number of outputs of~$\fxg$ is equal to the arity of~$f$.
\end{theorem}

The structure of the normal form is illustrated in Figure~\ref{nf},
page~\pageref{nf}.  Before proving this result we illustrate it with
an example.
\begin{example}
  Consider the function of Figure~\ref{ag} (page~\pageref{ag}).
  Recall that~$y$ and~$y'$ denote respectively the output of
  the entire system~$S$ and the output of~$f$.
  We get
  $$
  \left\{
  \begin{array}{lcl}
    \fxh(x_1,x_2,y') &=& m(y',q(x_1,x_2),p(q(x_1,x_2),x_2))   \\
    \fxg(x_1,x_2)   &=& \langle x_1,q(x_1,x_2)\rangle
  \end{array}
  \right.
  $$
  and, as stated in Theorem~\ref{nft} 
  we have $y=\fxh(\ov{x},f(\fxg(\ov{x})))$.
\end{example}

\begin{proof}
  Consider an acyclic graph that corresponds to an acyclic-PR-exp.  We describe
  an iterative algorithm that removes nodes of the
  graph until an acyclic-PR-exp with the form $\fxh(\ov{x},f(\fxg(\ov{x})))$ is
  obtained, see Figure~\ref{nf} (page~\pageref{nf}).  The output
  variable~$y$ of~$S$ must occur in a single equation (the last one)
  whose initial form is, say $y=H'(\ldots)$; the final form
  of~$H'$ will be denoted by~$\fxh$. Let the function~$f$ be
  represented by the equation~$y'=f(\ov{x'})$. This form of~$f$ will
  not change during the execution of the algorithm.

  With the exception of the input nodes~$\ov{x}$, of~$H'$, and
  of~$f$, the nodes of the graph are primitive recursive functions of
  the form~$u(\ov{w})$ that can be classified in~3 types\\
  \spac 
    -- \fINP: there is at least one path from the output 
    of~$u(\ov{w})$ to an input of~$f$.\\
  \spac 
    -- \fOUT: there is at least one path from the output~$y'$ of~$f$ 
    to one of the inputs~$\ov{w}$.\\
  \spac
    -- \NEITHER.\\
  In Figure~\ref{ag} (page~\pageref{ag}) $q$ is an \fINP\ node,
  $m$ is an \fOUT\ node, and~$p$ is a \NEITHER\ node.

  We emphasise that~$\ov{x}$ (the input nodes), $H'$, and~$f$ belong
  to neither of these classes and will never be removed ($H'$ may
  be modified during the normalisation procedure). Notice also that as
  the graph is acyclic, no node can be simultaneously of the \fINP\
  and \fOUT\ types.
  \begin{enumerate}
  \item [Step 1.]
  Remove all nodes~$u$ of the \fINP\ type that are not
  immediate predecessors of~$f$:\\
  The node~$u$ is removed and ``included'' in all the its
  successors (which may be of the \fINP, \fOUT or \NEITHER\ types) 
  (with the exception of~$f$); no loop is created in the graph. 
  This process may change the
  \NEITHER\ nodes and the nodes in the path~$y'$ to~$y$. In both cases
  the reason is that the output of~$u$ may connect to inputs of
  those nodes. \\
  In the end of the removal process, there is a single set of
  functions~$\fxg$ between the inputs~$\ov{x}$ of the system and the
  inputs of~$f$. If the output of a function~$g'$ of~$\fxg$ is also
  connected to an \fOUT\ or to a \NEITHER\ node, $g'$ is included in
  those nodes.
  \item [Step 2.]
  Remove all nodes~$u$ of the \fOUT\ type:\\
  In a similar way, these nodes are removed; \NEITHER\ nodes
  (but not \fINP nodes, now all included in~$\fxg$) can
  change during the process. Also, outputs of \NEITHER\ nodes
  (that were inputs of the \fOUT\ node~$u$) can now be inputs of~$H'$.
  \item [Step 3.]
  Remove all nodes~$u$ of the \NEITHER\ type:
  include them in the~$H'$ node. 
  \end{enumerate}

  In the end of the process, the final form of~$H'$ (that is, $\fxh$) 
  has~$\ov{x}$ and~$y$ as its only inputs.
\end{proof}

\subsubsection*{How to obtain the normal form from the 
  system of equations: an example}
The normal form algorithm can be rephrased in terms of the system
of equations that describes the system~$S$.
Instead of giving the algorithm in detail, we illustrate it
with the example of Figure~\ref{ag} (page~\pageref{ag}).  
The function~$f$, instance of the decision problem, is special;
it plays the role of an argument of the system.
$$
\begin{array}{lcl}
  \left\{
  \begin{array}{lcl}
    z &=& q(x_1,x_2)   \\
    y'&=& f(x_1,z)     \\
    v &=& p(z,x_2)     \\
    y &=& m(y',z,v)
  \end{array}
  \right.
&
\stackrel{(1)}{\longrightarrow}
&
  \left\{
  \begin{array}{lcl}
    z &=& q(x_1,x_2)            \\
    y'&=& f(x_1,z)              \\
    v &=& p(q(x_1,x_2),x_2)     \\
    y &=& m(y',q(x_1,x_2),v)
  \end{array}
  \right.
\\ \\
  \ 
&
\stackrel{(2)}{\longrightarrow}
&
  \left\{
  \begin{array}{lcl}
    z &=& q(x_1,x_2)   \\
    y'&=& f(x_1,z)     \\
    y &=& m(y',q(x_1,x_2),p(q(x_1,x_2),x_2))
  \end{array}
\right.
\end{array}
$$
The first system of equations corresponds directly to
Figure~\ref{ag}. Although not done here, the first step would be
to include all the inputs in the functions that have at least one
input as argument (with exception of~$f$). 

\begin{enumerate}
\item [Step 1.] The definition of~$z$ depends only on the
  inputs~$x_1$ and~$x_2$ (and not on ``intermediate''
  variables). Replace~$z$ by $q(x_1,x_2)$ in every line after its
  definition, except in the line that defines~$f$; for instance,
  $v=p(z,x_2)$ is replaced by $v=p(q(x_1,x_2),x_2)$.
  The definition of~$z$, namely $z=g(x_1,x_2)$ is not deleted
  because its output is an input of~$f$.

\item [Step 2.] The definition of~$v$ depends only on the
  inputs~$x_1$ and~$x_2$. Replace~$v$ by $p(q(x_1,x_2),x_2)$ in every
  line after its definition; in this case
  \begin{enumerate}[--]
  \item The definition of~$y$ becomes
    $y=m(y',q(x_1,x_2),p(q(x_1,x_2),x_2))$.
  \item  The definition of~$v$ is deleted.
  \end{enumerate}
\end{enumerate}

\subsubsection*{General case: a condition for the existence of a solution}
Using the normal form of~$S$, we see that
\begin{observation*}
  If~$S$ is a fixed acyclic PR structure (see Section~\ref{normalf}, 
  page~\pageref{normalf}), a primitive recursive function~$f$ is a 
  solution of the problem ``$\exists \ov{x}:S(f,\ov{x})=0$?''   iff
  the following condition holds
  $$
   \exists\ov{x},y': [\fxh(\ov{x},y')=0]\wedge [f(\fxg(\ov{x}))=y']\Mdot
  $$
\end{observation*}

Note however that, as already stated in page~\pageref{Sproblem}
(see also Figure~\ref{relat} in page~\pageref{relat}),
this observation {\em is not a ``closed form'' condition} 
(on the functions~$\fxg$ and~$\fxh$) for the existence of a solution
of ``$\exists \ov{x}:S(f,\ov{x})=0$''. Thus, it is not a
satisfactory solution of the problem.

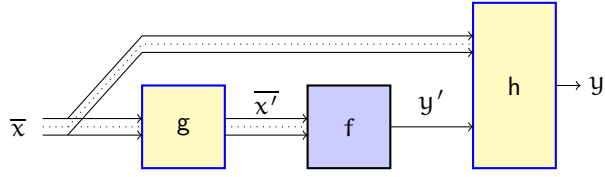
\begin{figure}[t]
  \begin{center}
  \begin{tikzpicture}[xshift=2.5cm, yshift=2.0cm, scale=2.2, auto,
    main node/.style={circle,fill=blue!20,draw}]
    \draw [blue, thick, fill=yellow!30] (0.75,0) rectangle (1.25,0.5);
    \draw [black, thick, fill=blue!20]  (1.75,0) rectangle (2.25,0.5);
    \draw [blue, thick, fill=yellow!30] (2.75,0) rectangle (3.25,1);
    \node at (0.0,0.25) {$\ov{x}$};
    \node at (3.5,0.5)  {$y$};
    \node at (1.0,0.25) {$\fxg$};
    \node at (2.0,0.25) {$f$};
    \node at (3.0,0.50) {$\fxh$};
    \node at (2.50,0.4) {$y'$};
    \node at (1.50,0.4) {$\ov{x'}$};
     \draw [->]     (3.25,0.5) -- (3.4,0.5);     
     \draw [->]     (2.25,0.25) -- (2.75,0.25);  
     \draw [->]     (1.25,0.20) -- (1.75,0.20);  
     \draw [->]     (1.25,0.30) -- (1.75,0.30);  
     \draw [dotted] (1.25,0.25) -- (1.75,0.25);  
     \draw [->]     (0.15,0.20) -- (0.75,0.20);  
     \draw [->]     (0.15,0.30) -- (0.75,0.30);  
     \draw [dotted] (0.15,0.25) -- (0.75,0.25);  
     \draw [->]     (0.30,0.20) -- (0.75,0.70) -- (2.75,0.70);  
     \draw [->]     (0.30,0.30) -- (0.75,0.80) -- (2.75,0.80);  
     \draw [dotted] (0.30,0.25) -- (0.75,0.75) -- (2.75,0.75);  
  \end{tikzpicture}
  \end{center}
  \caption{An arbitrary acyclic composition of one occurrence of~$f$ with
    fixed PR functions can be reduced to this ``normal form''. The
    corresponding expression has the form
    $y=\fxh(\ov{x},f(\fxg(\ov{x})))$
    where~$\fxg$ has multiple outputs.}
\label{nf}
\end{figure}

\begin{figure}[t]
  \begin{center}
  \begin{tikzpicture}[xscale=3.8, yscale=1.5, auto,
    main node/.style={circle,fill=blue!20,draw}]
    \draw [red!20, thin, fill=yellow!10] (1,1) ellipse (0.12 and 1.2);
    \node [circle,fill=yellow!20,draw=red!20,text width=11mm] (Z)  
          at (0.2,1.0)  {$\;y=0$};
    \node (y1) at (1.0,0.0)  {$y'$};
    \node (x)  at (1.0,2.0)  {$\ov{x}$};
    \node [circle,fill=yellow!20,draw=red!20,text width=9mm] (x1)  
          at (2.0,1.0)  {$\;\;\;\ov{x'}$};
    \node at (0.60,1.0)  {\textcolor{dbrown}{$\fxh^{-1}$}};
    \draw [red!90]    (0.40,1.00) -- (0.50,1.00);
    \draw [->,red!90] (0.70,1.00) -- (0.99,1.00);
    \draw [->,green!60!black] (x)  to node [swap]{
          \textcolor{black}{$\fxh$}} (Z); 
    \draw [->,brown!30!black] (y1) to node {$\fxh$} (Z); 
    \draw [->] (x)  to node {$\fxg$} (x1); 
    \draw [->, blue] (x1) to node {$f$} (y1);
    \draw [<->,dashed,brown!10!black] (x) to node {} (y1);
  \end{tikzpicture}
  \end{center}
  \caption{(Refer to Figure~\ref{nf}, page~\pageref{nf}.) 
    The global system~$S$, $y=S(f,\ov{x})$, has a~zero iff
    we have $f(\fxg(\ov{x}))=y'$ for some
    $(\ov{x},y')\in \fxh^{-1}(0)$. 
    For the particular case $S(f,\ov{x})\equiv \fxh(f(\fxg(\ov{x})))$,
    that is, when there is no \textcolor{dgreen}{green} arrow in the
    diagram,
    there is an explicit (``acyclic'') condition for the existence of a 
    zero of~$S$, see Theorem~\ref{ozHFG-t} (page~\pageref{ozHFG-t}).
    But no such condition is known for the general case.
   }
\label{relat}
\end{figure}
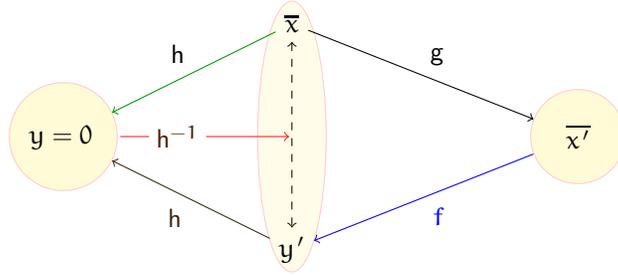

\subsection{Composition of~$f$ with itself}
\label{itself}
\subsubsection{Introduction}
Until now we have only studied PR decision problems in which
only one occurrence of~$f$ is allowed.  This is because we
view~$f$, the instance of the decision problem, as a part of a
larger system~$S$~-- an acyclic graph that containing a
node~$f$.

However, it is interesting to study systems containing 
more than one occurrence of~$f$, and in particular, systems in which the
composition of~$f$ with itself is allowed. 

{\em In this section we will only study a few problems of this kind.}
In particular we study the decidability of the problem
\begin{equation}
  \label{probFF}
  \text{``Given~$f$, does the function~$f(f(x))$ have a zero?''}
\end{equation}

\subsubsection{Solution of some problems in which
the primitive recursive function~$f$ is composed with itself}
\label{fff}
For the class of \ozN\ problems in which~$f$ may be composed with
itself, we do not know a general undecidability criterion. However,
we now study a few such problems.  Each of them turns out to be 
$\sigOne$-complete.

The problems that we will consider are the following.
\begin{enumerate}
\item [\ ] \FFN: given~$f$, does the function~$f(f(x))$ have at least
  one zero?''.
\item [\ ] \FFnN: given~$f$, does the function
  $\overbrace{f(f(\cdots f(x)))}^{\text{$n$ $f$'s}}$ have at least one zero?
  (\FFRN{2} is the problem \FFN.)
\item [\ ] \FFtZN: given~$f$, does the function~$f(f(x))$ have two or
  more zeros? 
\end{enumerate}

We will use the concept of ``graph of a function''.
\begin{definition}
  The graph of a function~$f$ is the directed graph
  $$
  (V,E)\;\;\; 
  \text{where $V=\ene$ and $E=\{(i,j):i\in\ene,\,f(i)=j\}$}
  $$
\close
\end{definition}

\begin{theorem}
  The problem \FFtZN\ is $\sigOne$-complete.
\end{theorem}

\begin{proof}
  Membership in~$\sigOne$ is obvious.
  To prove completeness reduce \ozN\ to \FFtZN. 
  Given an instance~$f$ of \ozN,
  define the function~$g$, instance of \FFtZN, as
  $$
  \begin{array}{lcll}
  g(0)    &=& 0    &                             \\
  g(1)    &=& f(0)    &                          \\
  g(2)    &=& f(1)    &                          \\
  g(3)    &=& f(2)    &                          \\
  \ldots&\ldots&\ldots
  \end{array}
  $$
  As $g(g(0))=g(0)=0$ the function~$g(g(x))$ has at least one
  zero. Also, if~$f(x)=0$ for some~$x$, $g(g(x))$ has two or more
  zeros, because in this case we have $g(g(x+1))=g(f(x))=g(0)=0$
  (and~$x+1\geq 1$).

  On the other hand, if~$f(x)$ has no zeros and~$x\geq 1$, we get
  $$
  g(g(x))=g(f(x-1))=g(y)=f(y-1)\geq 1
  $$
  where~$y\defined f(x-1)\geq 1$.
\end{proof}

\begin{theorem}
  The problem \FFN\ (page~\pageref{probFF}) is $\sigOne$-complete.
\end{theorem}

\begin{proof}
  Notice
  that~$\exists x:g(g(x))=0$
  iff the graph of~$g$ contains a path with length~2 ending in~0.

  We reduce \ozN\ to \FFN. Given an instance~$f$ of \ozN, define
  the graph of the function~$g$, the instance of \FFN, as:
  \begin{enumerate}
  \item [--] The set of nodes is~$\ene$.
  \item [--] For every pair $(i,j)$ in the graph of~$f$, that is, for
    every~$f(i)=j$ with~$i\in\ene$, there are two edges of the graph
    of~$g$: $(2i,2i+1)$ and $(2i+1,2j)$. 
    Notice that in the graph of~$g$ every odd numbered node has 
    indegree = outdegree = 1.
  \end{enumerate}
  Every node~$i$ of the graph of~$f$ corresponds to two
  nodes of the graph of~$g$: $2i$ and~$2i+1$. See Figure~\ref{FFtwo}
  (page~\pageref{FFtwo}) where the transformation is illustrated.

  Suppose that there is a solution of \ozN: in the graph of~$g$ the
  solution~$f(x)=0$ corresponds a path with length~2 ending in~0,
  namely $2x\to 2x+1 \to 0$. So, $g(g(2x))=0$.
  For example $f(2)=0$ and $g(g(4))=g(5)=0$.

  Suppose that there is a solution of \FFN: the graph of~$g$ has a path of
  length~2 ending in~0, say~$g(g(x))=0$; that path must have the
  form $2i\to 2i+1 \to 0$, because every even node~$2i$ only
  connects to the successor node~$2i+1$ (and never to~0). Thus, by
  the definition of~$g$, we must have~$f(i)=0$.
\end{proof}

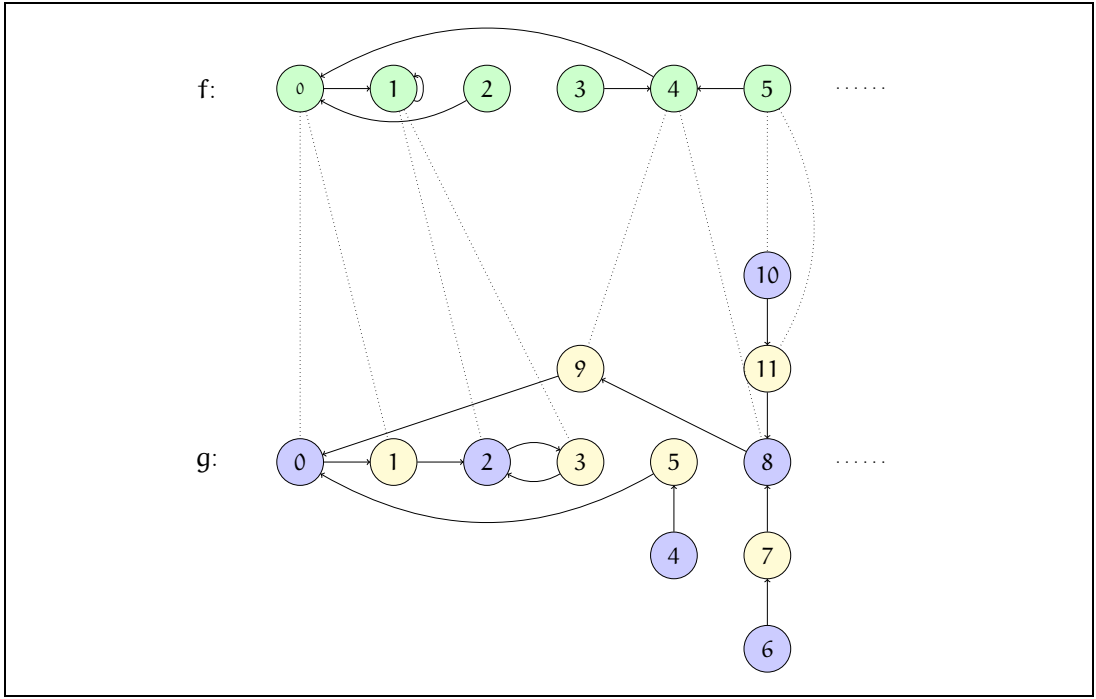
\begin{figure}
  \begin{center}
  \fbox{
  \begin{minipage}[c]{14 cm}
  \hfill
  \resizebox{95mm}{!}{%
  \begin{tikzpicture}[node distance=2.0cm,
     xshift=100mm, 
     main node/.style={circle,fill=blue!20,draw,minimum size=1.0cm},
     green/.style={circle,fill=green!20,draw,minimum size=1.0cm},
     others/.style={circle, draw=black, 
           fill=yellow!20, minimum size=1.0cm},
     empty/.style={circle, draw=white, 
           fill=white, minimum size=1.0cm},
     every loop/.style={max distance=3mm,in=30,out=-30}]
    \node[green] (0)              {$0$};
    \node[empty] (FF)[left  of=0 ]{\LARGE $f$:};
    \node[green] (1) [right of=0] {\Large $1$};
    \node[green] (2) [right of=1] {\Large $2$};
    \node[green] (3) [right of=2] {\Large $3$};
    \node[green] (E) [right of=3] {\Large $4$};
    \node[green] (F) [right of=E] {\Large $5$};
    \node[empty] (DOTS1) [right  of=F]{.\;.\;.\;.\;.\;.};

    \draw[->] (0) to node {$\ $} (1);
    \path[->] (1) edge [loop below] node {\ } (1);
    \draw[->] (3) to node {$\ $} (E);
    \draw[->] (F) to node {$\ $} (E);
    \draw[->, bend left] (2) to node {$\ $} (0);
    \draw[->, bend right] (E) to node {$\ $} (0);
    \node[empty]     (41)  [below of=0 ] {$\ $};
    \node[empty]     (42)  [below of=41] {$\ $};
    \node[empty]     (43)  [below of=42] {$\ $};
    \node[main node] (50)  [below of=43] {\Large $0$};
    \node[empty]     (GG)  [left  of=50] {\LARGE $g$:};
    \node[others]    (51)  [right of=50] {\Large $1$};
    \node[main node] (52)  [right of=51] {\Large $2$};
    \node[others]    (53)  [right of=52] {\Large $3$};
    \node[others]    (55)  [right of=53] {\Large $5$};
    \node[main node] (58)  [right of=55] {\Large $8$};
    \node[empty] (DOTS2)   [right of=58] {.\;.\;.\;.\;.\;.};
    \node[main node] (54)  [below of=55] {\Large $4$};
    \node[others]    (57)  [below of=58] {\Large $7$};
    \node[main node] (56)  [below of=57] {\Large $6$};
    \node[others]    (59)  [above of=53] {\Large $9$};
    \node[others]    (511) [above of=58] {\Large $11$};
    \node[main node] (510) [above of=511]{\Large $10$};
    \draw[dotted]  (0)  to node {$\ $} (50);
    \draw[dotted]  (0)  to node {$\ $} (51);
    \draw[dotted]  (1)  to node {$\ $} (52);
    \draw[dotted]  (1)  to node {$\ $} (53);
    \draw[dotted]  (E)  to node {$\ $} (58);
    \draw[dotted]  (E)  to node {$\ $} (59);
    \draw[dotted]  (F)  to node {$\ $} (510);
    \draw[dotted, bend left] (F) to node {$\ $} (511);
    \draw[->]  (50) to node {$\ $} (51);
    \draw[->]  (51) to node {$\ $} (52);
    \draw[->]  (56) to node {$\ $} (57);
    \draw[->]  (57) to node {$\ $} (58);
    \draw[->]  (54) to node {$\ $} (55);
    \draw[->] (510) to node {$\ $} (511);
    \draw[->] (511) to node {$\ $} (58);
    \draw[->] (58) to node {$\ $} (59);
    \draw[->] (59) to node {$\ $} (50);
    \draw[->, bend left] (52) to node {$\ $} (53);
    \draw[->, bend left] (53) to node {$\ $} (52);
    \draw[->, bend left] (55) to node {$\ $} (50);
  \end{tikzpicture}
  }
  \hfill\ 
  \rule[-2mm]{0.0mm}{5mm}
  \end{minipage}
  }
  \end{center}
  \caption{Example of the reduction from \ozN\ (function~$f$) to \FFN\ 
    (function~$g$).
    The upper diagram (green nodes) represents part of the
    graph of~$f$ with~$f(0)=1$, $f(1)=1$, $f(2)=0$,
    $f(3)=4$, $f(4)=0$, and~$f(5)=4$.
    The bottom diagram represents the transformed graph, or graph of the
    function~$g$. Each node~$i$ of the top diagram is ``transformed'' in
    two nodes, $2i$ (blue) and $2i+1$ (yellow). 
    Dotted lines connect the upper nodes~0, 1, 4, and~5 to the 
    corresponding pairs of nodes in the bottom diagram (connections from
    upper nodes~2 and~3 are not represented).
    Each edge $(i,j)$ of~$f$ (top diagram) is mapped in two edges 
    of~$g$ (bottom diagram):
    $(2i,2i+1)$ and $(2i+1,2j)$.} 
\label{FFtwo}
\end{figure}

The previous proof can be easily generalised.
\begin{theorem}
  For every~$n\geq 1$ the problem \FFnN\ is $\sigOne$-complete.
\end{theorem}

\section{Comparing the degree of undecidability: conjectures}
\label{PR-ParRec}
We compare in the general case the degree of 
undecidability of a ParRec problem and of the corresponding PR problem.
ParRec indices will be denoted by~$e$ and PR indices by~$p$.
The set of all~$p$, that is, of the indices 
that represent Loop programs, will be denoted by~$\IPR$
(the set $\IPR\subset\ene$ is recursive).

We first characterise the PR problem that corresponds to a given ParRec 
problem.
\begin{definition}[Problem correspondence]
  \label{correspondence}
  Let a decision problem about ParRec functions be
  ``given~$e\in\ene$, is $\stat(\vph_e)$?'',
  where $\stat(f)$ denotes some statement about the function~$f$.
  The {\em corresponding} PR decision problem is
  ``given~$p\in\IPR$, is $\stat(\vph_p)$?''. \close
\end{definition}

Thus, the correspondence between the two kinds of problems is simply 
the restriction of the index set, $\ene\to\IPR$.

\begin{definition}
  If~$n$ is a positive integer, the decision problem~$P$
  is {\em located} at level~$n$ of the AH if $P\in\sigN\cup\piN$
  but $P\not\in\deltaN$.
  If $P\in\sigN\setminus\piN$ we say that~$P$ is {\em located} in~$\sigN$
  and similarly, if $P\in\piN\setminus\sigN$ we say that~$P$ is 
  {\em located} in~$\piN$.
  If a decision problem~$P$ is $\sigN$-complete or $\piN$-complete, then~$P$
  is located at level~$n$.
  For $n\geq 2$ we say that a set in 
  $\deltaN\setminus(\sigNO\cup\piNO)$ is {\em located}
  at level $n-\sfrac{1}{2}$ of the AH.
  If a PR decision problem is located at level~$m$ of the AH and
  the corresponding ParRec problem is located at level~$n$ of the AH
  we say that the {\em undecidability jump} (or JUMP) is $[m\!\to\! n]$.
  The {\em undecidability discrepancy} (or DISCREP) between those 
  problems is $\abs{n-m}$.
\close
\end{definition}
It is a consequence of the Arithmetic Hierarchy Theorem (see for instance
\cite[Theorem IV.1.13]{odi}) that for every positive integer~$n$ there exist
decision problems located at level~$n$ and decision problems located at level 
$n+\sfrac{1}{2}$.\\
A {\em partition} of the arithmetical predicates~(\cite{hermes}) is
$$
   \deltaOne,\, \pc{\sigOne},\, \pc{\piOne},\,
   [\deltaTwo\setminus(\sigOne\cup\piOne)],\, \pc{\sigTwo},\, \pc{\piTwo},\,
   [\deltaThree\setminus(\sigTwo\cup\piTwo)],\ldots
$$
where $\pc{X}$ means the class of problems complete in~$X$.
\bigskip

Theorem~\ref{alsoPART} (page~\pageref{alsoPART}) is a simple
consequence of the fact that every PR index is also a ParRec index,
that is, $\IPR\subset\ene$. We now show that
the index restriction $\ene\to\IPR$ {\em does not increase} the AH level in
which a decision problem is located.

\begin{theorem}
  \label{reduce}
  Let~$P$ be a ParRec decision problem and let~$\PRversion{P}$ be the 
  corresponding PR version. If~$P$ and~$\PRversion{P}$ are located 
  respectively at the levels~$n$ and~$m$ (possibly a half-integers) of the 
  arithmetic hierarchy, then $n\geq m$.
\end{theorem}

\begin{proof}
  Suppose first that~$P$ is located in~$\sigN$ or in~$\piN$.
  Let~$P$ and~$\PRversion{P}$ be expressed as $\stat(\vph_e)$ and
  $\stat(\vph_p)$, respectively; see Definition~\ref{correspondence}, 
  page~\pageref{correspondence}. The two statements are identical and
  can be written in the same Tarski-Kuratowski normal form
  (see for instance \cite[Chapter 14]{Rogers}).
  It follows that~$\PRversion{P}$ is also at level~$n$ (but possibly
  not located at level~$n$). The level~$m$
  at which~$\PRversion{P}$ is located must then satisfy $m\leq n$.

  If~$P$ is located in~$\deltaM$, the corresponding statement has
  the form $\stat(\vph_e)\wedge\stat'(\vph_e)$, where one conjunct
  corresponds to~$\sigM$ and the other to~$\piM$. The
  problem~$\PRversion{P}$ also satisfies
  $\stat(\vph_p)\wedge\stat'(\vph_p)$,
  and we may reason as above to show that $m\leq n$.
\end{proof}

\label{conjecture}

The restriction $\IPR\to \ene$ may reduce the degree of undecidability of a 
problem, see Figure~\ref{c-summary} (page~\pageref{c-summary}) and 
Figure~\ref{changes} (page~\pageref{changes}).
After analysing the JUMP of a few decision problems we will conjecture that 
all non-negative JUMPs and every non-negative DISCREPs are possible.


\begin{landscape}
\begin{figure}
$$
\arraycolsep=2.5mm
\renewcommand{\arraystretch}{1.7}
\setlength{\arraycolsep}{10pt}
\begin{array}{p{2.0mm}|ccccccc|}   
   \multicolumn{1}{c}{} &
   \multicolumn{1}{c}{\text{problem}}   &   
   \text{function}   & \text{instance} &
   \text{statement} &\text{AH class} &
   \text{JUMP}      & \multicolumn{1}{c}{\text{DISCREP}}    \\ \cline{2-8}
   \multirow{2}{10mm}{1}&
   \multirow{2}{*}{$\exists x\!:\!f(x)\!=\!1$} &
   \text{PR}         & p & \exists x:\vph_p(x)=1 & 
                       \text{$\sigOne$-complete}    &
                   \multirow{2}{*}{$[1\!\to\!1]$} & 
                   \multirow{2}{*}{0}  \\
   &&
   \text{ParRec}     & e & \exists x:\vph_e(x)=1
                     & \text{$\sigOne$-complete}   &&
   \\ \cline{2-8}
   \multirow{2}{10mm}{2}&
   \multirow{2}{*}{{\GkCodL = 1}} &
   \text{PR}         & p & \lvert\cod{\vph_p}\rvert=1 & 
                       \text{$\piOne$-complete}     &
                   \multirow{2}{*}{$[1\!\to\!\sfrac{3}{2}]$} & 
                   \multirow{2}{*}{$\sfrac{1}{2}$}  \\
   &&
   \text{ParRec}     & e &  \lvert\cod{\vph_e}\rvert=1 
                     & \text{$\deltaTwo\!\setminus\!(\sigOne\cup\piOne)$} &&
   \\ \cline{2-8}
   \multirow{2}{10mm}{3}&
   \multirow{2}{*}{\GeozN} &
   \text{PR}         & p & \exists! x:\vph_p(x)=0 & 
                       \text{$\deltaTwo\!\setminus\!(\sigOne\cup\piOne)$}    &
                   \multirow{2}{*}{$[\sfrac{3}{2}\!\to\!2]$} & 
                   \multirow{2}{*}{$\sfrac{1}{2}$}  \\
   &&
   \text{ParRec}     & e & \exists! x:\vph_e(x)=0
                     & \text{$\piTwo$-complete}   &&
   \\ \cline{2-8}
   \multirow{2}{10mm}{4}&
   \multirow{2}{*}{HP} &
   \text{PR}         & \ang{p,x} & \exists t:T(p,x,t)=0 & 
                   \text{recursive}               &
                   \multirow{2}{*}{$[0\!\to\!1]$} & 
                   \multirow{2}{*}{1}  \\
   &&
   \text{ParRec}     & \ang{e,x} & \exists t:T(e,x,t)=0
                     & \text{$\sigOne$-complete}    &&
   \\ \cline{2-8}
   \multirow{2}{10mm}{5}&
   \multirow{2}{*}{TOTAL} &
   \text{PR}         & p & \forall x\:\exists t:T(p,x,t)=0
                     & \text{recursive}           &
                   \multirow{2}{*}{$[0\!\to\!2]$} & 
                   \multirow{2}{*}{2}  \\
   &&
   \text{ParRec}     & e & \forall x\:\exists t:T(e,x,t)=0
                     & \text{$\piTwo$-complete}    &&
   \\ \cline{2-8}
   \multirow{2}{10mm}{6}&
   \multirow{2}{*}{{\recurN}} &
   \text{PR}         & p & \text{$W_p$ is recursive}
                     & \text{recursive}           &
                   \multirow{2}{*}{$[0\!\to\!3]$} & 
                   \multirow{2}{*}{3}  \\
   &&
   \text{ParRec}     & e & \text{$W_e$ is recursive}
                     & \text{$\sigThree$-complete}    &&
   \\ \cline{2-8}
\end{array}
$$
\caption{The degree of undecidability of six decision problems is compared.
  EOZ means \GeozN.
  For each problem two cases are considered: the instance is a partial 
  recursive (ParRec) function and the instance is a primitive recursive (PR) 
  function.}
\label{four}
\end{figure}
\end{landscape}

\subsection*{Some decision problems: discrepancy}

\begin{problem}
\label{PfEQone}
\noindent ``$\exists x:f(x)=1$?''. 
Line~1 of Figure~\ref{four}, page~\pageref{four}.\\
In detail the ParRec statement is
$
   \exists x\:\exists t:(T(e,x,t)=0) \wedge (U(e,x,t_0)=1)
$,
where~$t_0$ is a value of~$t$ satisfying $T(e,x,t_0)=0$ and
the notation of Theorem~\ref{Knft} (page~\pageref{Knft}) is used.\closex
\end{problem}

\begin{problem}
\label{PcodONE}
\noindent ``\GkCodON?''. 
Line~2 of Figure~\ref{four}, page~\pageref{four}.
See item~\ref{proof-COD-one-ParR}, page~\pageref{proof-COD-one-ParR}.\closex
\end{problem}

\begin{problem}
\noindent \Geoz. 
Line~3 of Figure~\ref{four}, page~\pageref{four}.
See item~\ref{eozIT} (page~\pageref{eozIT}) and Theorem~\ref{t-eoz} 
(page~\pageref{t-eoz}).\closex
\end{problem}

\begin{problem}
\noindent{HP, the halting problem.}
Line~4 of Figure~\ref{four}, page~\pageref{four}.
We use the notation explained in Definition~\ref{interval} 
(page~\pageref{interval}).
$\ang{e,x}$ and $\ang{p,x}$ denote respectively bijections 
$\ene\times\ene\to\ene$ and $\IPR\times\ene\to\ene$.\closex
\end{problem}

\begin{problem}
\noindent{TOTAL problem.}
Line~5 of Figure~\ref{four}, page~\pageref{four}.\closex
\end{problem}

\begin{problem}
\label{Precur}
\noindent{\recurN\ problem.} 
Line~6 of Figure~\ref{four}, page~\pageref{four}.\\
Let \recurN\ be the set defined in \cite[Definition~4.14 (page~21)]{soare},
$
\recur= \{e \mid \text{$W_e$ is recursive}\}
$.
This set is $\sigThree$-complete, see for instance
\cite[\S 14.8, Theorem~XVI (page 327)]{Rogers} and
\cite[Corollary~3.5 (page 66)]{soare}.
The PR version of this problem is
$
\recurPR = \{p \mid \text{$W_p$ is recursive}\}
$.
Once $\forall p\in\IPR:W_p=\ene$, we have \recurPRN~=~$\IPR$, a recursive 
set.\closex
\end{problem}

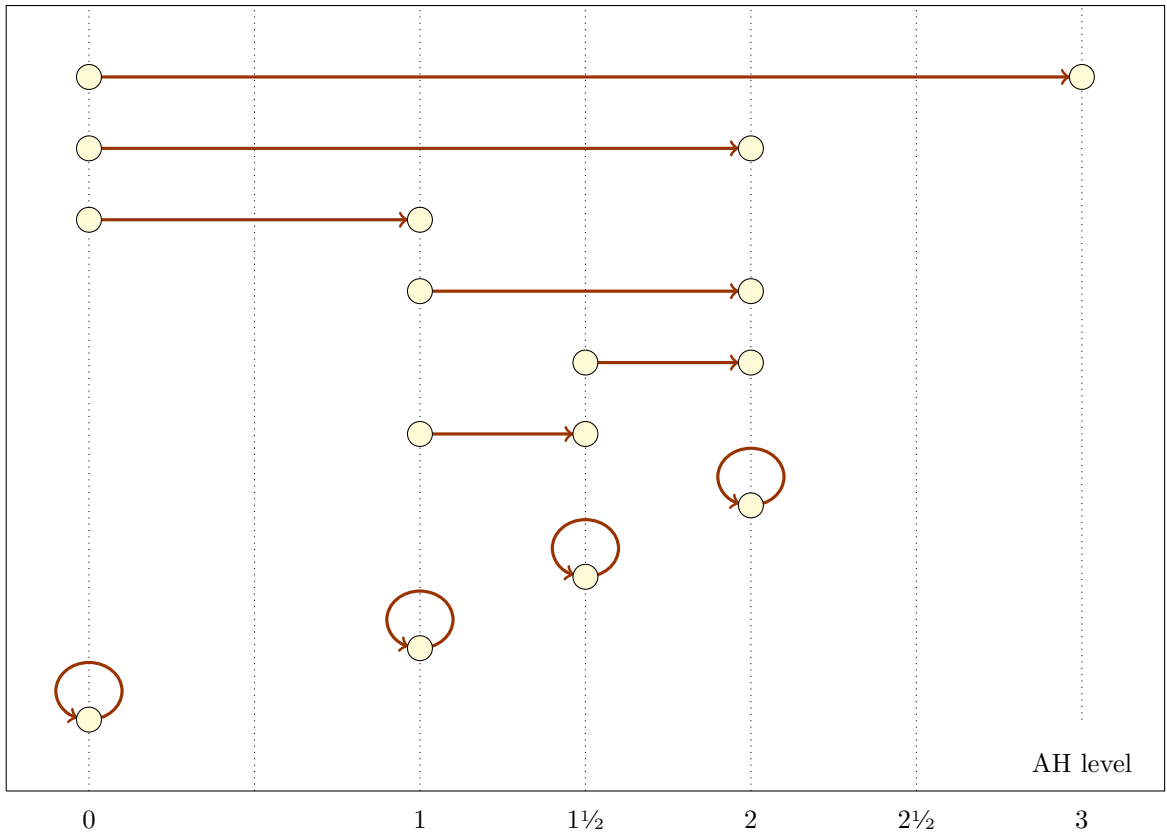
\begin{figure}
\begin{center}
  \begin{tikzpicture}[xshift=0.5cm, yshift=2.0cm, 
    xscale=2.2, yscale=1.9,auto]
  \tikzset{others/.style={circle, draw=black, 
           fill=yellow!20, minimum size=2mm},
           double/.style={circle, draw=black, very thick, 
           fill=orange!30, minimum size=2mm}
   }
    \node at (0.0,-0.2)     {0};
    \node at (2.0,-0.2)     {1};
    \node at (3.0,-0.2)     {1\sfrac{1}{2}};
    \node at (4.0,-0.2)     {2};
    \node at (5.0,-0.2)     {2\sfrac{1}{2}};
    \node at (6.0,-0.2)     {3};
    \node at (6.0, 0.2)     {\text{AH level}};
  \draw[black, thin]  (-0.5,0)--(6.5,0)--(6.5,5.5)--(-0.5,5.5)--(-0.5,0);
  \draw[black, thin, dotted]    (0,0) -- (0,5.5);
  \draw[black, thin, dotted]    (1,0) -- (1,5.5);
  \draw[black, thin, dotted]    (2,0) -- (2,5.5);
  \draw[black, thin, dotted]    (3,0) -- (3,5.5);
  \draw[black, thin, dotted]    (4,0) -- (4,5.5);
  \draw[black, thin, dotted]    (5,0) -- (5,5.5);
  \draw[black, thin, dotted]    (6,0.5)--(6,5.5);
    \draw[very thick,->,dred] (0.00,0.50) arc (-90:250:2.00mm);
    \draw[very thick,->,dred] (2.00,1.00) arc (-90:250:2.00mm);
    \draw[very thick,->,dred] (3.00,1.50) arc (-90:250:2.00mm);
    \draw[very thick,->,dred] (4.00,2.00) arc (-90:250:2.00mm);
    \node[others] at (0.00,0.50)   {};
    \node[others] at (2.00,1.00)   {};
    \node[others] at (3.00,1.50)   {};
    \node[others] at (4.00,2.00)   {};
    \draw[very thick,->>,bend right,dred] (2.0,2.5)--(3.0,2.5);
    \node[others] at (2.00,2.50)   {};
    \node[others] at (3.00,2.50)   {};
    \draw[very thick,->>,dred] (3.0,3.0)--(4.0,3.0);
    \node[others] at (3.00,3.00)   {};
    \node[others] at (4.00,3.00)   {};
    \draw[very thick,->>,dred] (2.0,3.5)--(4.0,3.5);
    \node[others] at (2.00,3.50)   {};
    \node[others] at (4.00,3.50)   {};
    \draw[very thick,->>,dred] (0.0,4.0)--(2.0,4.0);
    \node[others] at (0.00,4.00)   {};
    \node[others] at (2.00,4.00)   {};
    \draw[very thick,->>,dred] (0.0,4.5)--(4.0,4.5);
    \node[others] at (0.00,4.50)   {};
    \node[others] at (4.00,4.50)   {};
    \draw[very thick,->>,dred] (0.0,5.0)--(6.0,5.0);
    \node[others] at (0.00,5.00)   {};
    \node[others] at (6.00,5.00)   {};
  \end{tikzpicture}
  \end{center} 
  \caption{Levels in the arithmetic hierarchy of the PR (left) and ParRec (right)
    versions of 10 decision problems, including  Examples~\ref{PfEQone} 
    (page~\pageref{PfEQone}) to~\ref{Precur} (page~\pageref{Precur})
    and other problems studied in this work.
    In the four lower decision problems the PR and ParRec levels are identical. 
    All decision problems (in both versions) located in integer AH levels are 
    complete in the corresponding class.
    (The level~$\sfrac{1}{2}$ does not exist because 
     $\sigZero=\piZero=\deltaZero=\deltaOne$.)
}
\label{discreps}
\end{figure}

\subsection*{A conjecture}
Looking to Figure~\ref{discreps} (page~\pageref{discreps}), to which
other examples could easily be added, the following conjecture seems natural.
\begin{conjecture}
  \label{distance}
  For every pair $n\geq m\geq 0$ with 
  $n,\,m\in\{0,\,1,\,\sfrac{3}{2},\,2,\,\sfrac{5}{2},\,3,\,\sfrac{7}{2}\ldots\}$
  there is a
  ParRec decision problem located at level~$n$ of the arithmetic hierarchy
  such that the corresponding PR decision problem is located at level~$m$,
  that is, with {\rm JUMP}~=~$[m\!\to\!n]$
  (the case $m>n$ is impossible by Theorem~\ref{reduce}, page~\pageref{reduce}).
\end{conjecture}

If this conjecture is true, the undecidability discrepancy between a 
{\rm ParRec} problem and the corresponding {\rm PR} can be any
non-negative integer and any half-integer greater than~1. 

\subsection{Positive discrepancy: why?}
\label{reasons}
We have seen that some decision problems have DISCREP$>$0. There are
several reasons for this. For instance
\begin{enumerate}[(1)]
\item \label{Rtotal}
  {\em PR functions are total.}\\
  That is, $\forall p\forall x\exists t:T(p,x,t)$, using the notation of
  Definition~\ref{interval}, page~\ref{interval}.
  This the reason that applies to all problems with DISCREP$>0$
  that we studied so far.
\item  \label{Rbound}
  {\em There are recursive functions that grow ``too fast''.}\\
  There are recursive functions that are asymptotic upper bounds of 
  every PR function. As a consequence, there are decision problems that
  are trivial for PR functions but undecidable for ParRec functions.
  One example is $g(n)\defined A(n,n)$ where~$A$ is the 
  Ackermann function as defined in~\cite{calude} 
  or~\cite{hermes}\footnote{A similar function $2\!\stackrel{n}{\uparrow}\!n$ 
    where the Knuth superpower notation is used, see for instance 
    \cite{knuthCoping,MP80}.}; 
  The decision problem\footnote{Where ``$f(n)<g(n)$'' means that $f(n)$ is 
    defined and its value is less than~$g(n)$.} 
  ``given~$f$, $\exists n_0\forall n\geq n_0:f(n)<g(n)$?'' is trivial for 
  PR functions but undecidable for ParRec functions.
\item  \label{RzeroOne}
  {\em There are also ``small'' recursive functions that are not PR.}\\
  There are recursive functions bounded by a constant that are not PR.
  An example is the function:
  $$
  h(p)=
  \left\{
  \begin{array}{ll}
    0 & \spac\text{if $p\in\IPR$ and $\vph_p(p)\neq 0$}\\
    1 & \spac\text{otherwise}\Mdot
  \end{array}
  \right.
  $$  
  The function~$h(p)$ has codomain $\{0,1\}$, is recursive, and not PR. 
  However, the ParRec decision
  problem ``given~$e$, is $\vph_e=h$?'' is undecidable.
\item  \label{Rindex}
  {\em The PR index gives information about the function.}\\
  From the index~$p$ of a PR function~$\vph_p$ we can obtain some
  information about the function~$\vph_p$, for instance the
  PR-complexity of~$\vph_p$, see~\cite{hoyrup,hoyrup-rojas}.
  But no information about a ParRec function~$f$ can be obtained from
  a corresponding index (Rice Theorem, see for instance~\cite{Rogers,BBJ,davis}).
\end{enumerate}

We see that there are several very different reasons for a positive discrepancy.
Any definition of a class of total recursive functions has as a consequence
that some properties hold for all members of the class; for instance (1)-(4) 
hold for PR functions. These properties may be undecidable for ParRec functions
but trivial for PR functions. This fact leads us to conjecture the following.
\begin{conjecture}
  \label{discrep}
  There is no algorithm that, given a decision problem about a
  function, computes the discrepancy between the
  ParRec and the PR versions of that problem.
\end{conjecture}



\section{Conclusions and open problems}
\label{concs}
The primitive recursive (PR) functions are an important subclass
of the recursive (total computable) functions.
Many interesting decision problems related to PR
functions are undecidable. In this work we studied the degree of
undecidability of a relatively large number of these problems.  
All these problems are either m-complete in the
corresponding class of the arithmetic hierarchy or belong to
the class $\deltaTwo\setminus(\sigOne\cup\piOne)$. Their degree of
undecidability was compared with the corresponding 
decision problems associated with partial recursive (ParRec) functions,
as exemplified in Figure~\ref{changes}, page~\pageref{changes}.

The exact conditions for the decidability of the general decision
problem ``$\exists \ov{x}:\;\fxh(f(\fxg(\ov{x})))=0?$, where~$\fxg$
and~$\fxh$ are fixed PR functions ($\fxg$ may have 
multiple outputs) and~$f$ is the instance of the problem (a PR function), 
have been established.

A more general situation was also studied:
the instance of the problem, a PR function~$f$, is
part of an arbitrary PR acyclic graph~$S$ (having~$f$ and fixed primitive 
recursive functions as nodes).  Let
the corresponding function be~$S(f,\ov{x})$; the question is
``$\exists x:\;S(f,\ov{x})=0?$''. A normal form for these acyclic
function graphs was obtained, namely 
$\fxh(\ov{x},f(\fxg(\ov{x}))$; in contrast with the
$\fxh(f(\fxg(\ov{x}))$ class mentioned above, we conjecture that there
is no {\em closed} necessary and sufficient condition for decidability.

Another generalisation was were briefly studied: problems in which 
the instance~$f$ can be compose with itself.

\bigskip

Of course, many problems remain open. Some of the more interesting are:
\begin{enumerate}[--]
\item Clarify the relation between PR and ParRec functions, namely
\begin{enumerate}[(i)]
 \item Given any decision problem about a function, relate
   the degree of undecidability in two cases:
   when the function is PR and when the function is ParRec.
   In Section~\ref{PR-ParRec} (Conjecture~\ref{discrep}, 
   page~\pageref{discrep}) we conjecture that this may be not possible.
 \item
   Is every undecidability discrepancy possible? Proof or disproof the
   Conjecture~\ref{distance} (page~\pageref{distance}).
   We saw decision problems for which the undecidability discrepancy is~0, 
   $\sfrac{1}{2}$, 1, 2, and~3,
   see Section~\ref{PR-ParRec} (page~\pageref{PR-ParRec}) and
   Figure~\ref{discreps} (page~\pageref{discreps}).
 \item Study decision problems whose positive undecidability discrepancy
   is due to other reasons (besides the fact that every PR function is total); 
   see the discussion in Section~\ref{PR-ParRec}, page~\pageref{PR-ParRec}.
\end{enumerate}
\item Find closed decidability conditions for the general
  system~$S(f,\ov{x})$ described above. 
\item [--] Study in general the decidability problems in which several
  occurrences of the instance~$f$ may occur; this includes expressions in 
  which~$f$ composes with itself.
\end{enumerate}

\section*{Acknowledgements}
{\footnotesize
The author wants to thank Noson Yanofsky of CUNY University
for his helpful comments and for encouraging the pursuit of the research 
on primitive recursive functions.  
}

\nocite{shoenfieldBOOK}

\bibliographystyle{plain}
\bibliography{/Users/acm/RESEARCH/reversible/pr/bib.bib}

\end{document}